\documentclass[11pt]{article}
\usepackage{palatino}
\usepackage{pgfplots} 
\usepackage{tikz, tikzpeople}
\usetikzlibrary{calc,intersections,decorations.pathmorphing,decorations.markings,patterns,arrows.meta,positioning}
\usetikzlibrary{angles,quotes}
\usetikzlibrary{arrows.meta}
\usetikzlibrary{arrows}
\usetikzlibrary{shapes.multipart}
\usetikzlibrary{fit}
\usetikzlibrary{shapes.geometric,positioning}
\usetikzlibrary{calc}
\usepackage{bm}
\usepackage[utf8]{inputenc}
\usepackage{mathtools}
\usepackage{xfrac}
\usepackage{graphicx}
\usepackage[linesnumbered,ruled,vlined]{algorithm2e}
\usepackage{cite}
\usepackage{amsmath,amssymb,amsfonts, amsthm,  mathrsfs}
\usepackage{algpseudocode}
\usepackage{textcomp}
\usepackage{xcolor}
\usepackage{braket}
\usepackage{enumitem}
\usepackage{cancel}
\usepackage{hyperref}
\usepackage{stfloats}

\usetikzlibrary{patterns}
\newtheorem{lemma}{Lemma}
\newtheorem{theorem}{Theorem}

\newtheorem{remark}{Remark}
\newtheorem{corollary}{Corollary}
\newtheorem{definition}{Definition}

\newlength{\leftstackrelawd}
\newlength{\leftstackrelbwd}
\def\leftstackrel#1#2{\settowidth{\leftstackrelawd}%
{${{}^{#1}}$}\settowidth{\leftstackrelbwd}{$#2$}%
\addtolength{\leftstackrelawd}{-\leftstackrelbwd}%
\leavevmode\ifthenelse{\lengthtest{\leftstackrelawd>0pt}}%
{\kern-.5\leftstackrelawd}{}\mathrel{\mathop{#2}\limits^{#1}}}

\DeclareMathOperator{\Tr}{Tr}
\DeclareMathOperator{\supp}{supp}
\DeclareMathOperator*{\argmax}{arg\,max}

\newcommand{\mN}{\mathsf{N}} 
 
\newcommand{\mP}{\mathsf{P}}

\newcommand{\unif}{\mathsf{Unif}}

\newcommand{\C}{\mathchoice
  {\mathrm{\scriptscriptstyle C}} 
  {\mathrm{\scriptscriptstyle C}} 
  {\mathrm{\scriptscriptstyle C}} 
  {\mathrm{\scriptscriptstyle C}} 
}

\newcommand{\EA}{\mathchoice
  {\mathrm{\scriptscriptstyle EA}} 
  {\mathrm{\scriptscriptstyle EA}} 
  {\mathrm{\scriptscriptstyle EA}} 
  {\mathrm{\scriptscriptstyle EA}} 
}

\newcommand{\opt}{\mathchoice
  {\mathrm{\scriptscriptstyle opt}} 
  {\mathrm{\scriptscriptstyle opt}} 
  {\mathrm{\scriptscriptstyle opt}} 
  {\mathrm{\scriptscriptstyle opt}} 
}

\allowdisplaybreaks

\title{Quantum Entanglement Assistance Improves the Capacity and Activates the Zero-Error Capacity of Classical Channels with Causal CSIT}

\author{Yuhang Yao, Syed A. Jafar\\
{\small Department of Electrical Engineering and Computer Science}\\
{\small University of California Irvine, Irvine, CA 92697}\\
{\small \it Email: \{yuhangy5, syed\}@uci.edu}
}

\date{}

\begin{document}
\maketitle

\begin{abstract}
For  classical point-to-point channels, it has been shown by Bennett et al. that quantum entanglement assistance cannot improve their capacity, and by Cubitt et al. that entanglement assistance cannot activate (increase from zero to non-zero) their zero-error capacity. In contrast, we show that for classical point-to-point channels with causal CSIT (channel state information at the transmitter), quantum entanglement assistance can in some cases improve their capacity, and in some cases activate their zero-error capacity.
\end{abstract}

\section{Introduction}
The fundamental limits of reliable communication are among the most essential questions in information theory. The \emph{capacity} of a classical communication channel represents such a limit, defining a threshold for the rate of reliable communication which cannot be surpassed even if the encoder and decoder are equipped with unlimited (classical) computational resources, and are allowed to share unlimited amounts of (classical) common randomness in advance. The standard notion of capacity --- henceforth referred to as \emph{classical} capacity --- does not, however, account for scenarios in which the encoder and decoder  also share prior \emph{quantum} resources, such as \emph{entanglement.} 
\subsubsection*{Entanglement-assisted Communication}
On one hand, quantum resources are non-signaling --- quantum entanglement shared between two parties does not by itself allow them to communicate with each other \cite{PRbox}. On the other hand, however, shared quantum resources can assist in the encoding and decoding of messages, e.g., through local measurements conditioned on the information available to the sender and receiver. The outcomes of such measurements can exploit \emph{non-local correlations} \cite{barrett2005nonlocal} that are enabled by entanglement, and are unattainable by classical means. Thus, entanglement-assisted communication has the potential to surpass the classical capacity barrier. Indeed, the search for settings \emph{where entanglement-assisted capacity exceeds  classical capacity} is an active area of research, in pursuit of this particular form of \emph{quantum advantage}. This will also be the motivation for our work in this paper.

\subsubsection*{The Search for a Quantum Advantage}
Despite the evident potential, such a quantum advantage has not been easy to find, especially in the point-to-point communication setting. In fact in many cases it has been shown that such an advantage does not exist. Notably, Bennett et al. showed in \cite{Bennett_Shor_Smolin_Thapliyal_PRL} that entanglement assistance cannot increase the classical capacity of a point-to-point channel. Matthews in \cite{matthews2012linear} further strengthened this result, showing that any non-signaling resource (which strictly includes quantum entanglement) cannot improve the classical capacity of a point-to-point channel. On the positive side, significant advantages from entanglement-assistance have been found for a more restricted notion of capacity --- the zero-error capacity \cite{cubitt2010improving, leung2012entanglement, cubitt2011zero}. However, a sobering observation is that entanglement-assistance cannot \emph{activate} (i.e., increase from zero to non-zero) the zero-error capacity of a point-to-point channel \cite{cubitt2010improving}. Finally, while our focus in this work is on point-to-point communication, it is noteworthy that the scope of search for capacity improvements from entanglement-assistance has also been extended to communication \emph{networks} \cite{Quek_Shor, leditzky2020playing, seshadri2023separation,pereg2024MAC_QEassist, hawellek2025interference}, leading to various discoveries that include, e.g.,  an approximately $5\%$ \cite[Cor. 7]{seshadri2023separation} capacity gain from entanglement assistance in a $2$-to-$1$ classical multiple access channel.

\subsubsection*{Classical Channel with State and Causal CSIT}
Among the canonical models of point-to-point communication in information theory, is the `channel with state' setting, defined by a tuple $(\mathcal{X},\mathcal{Y},\mathcal{S},\mP_S, \mN_{Y|XS})$, such that $\mathcal{X},\mathcal{Y},\mathcal{S}$ represent the alphabet sets for the input, output and channel state, $\mP_S$ is the distribution of the channel state (assumed i.i.d. varying across channel-uses), and $\mN_{Y|XS}$ defines the output distribution of the memoryless channel conditioned on its current input and current state. Depending on whether the \underline{c}hannel \underline{s}tate \underline{i}nformation at the \underline{t}ransmitter (CSIT) is made available in a causal or non-causal fashion, the setting is identified as a \emph{causal CSIT} or a \emph{non-causal CSIT} setting, respectively. The classical capacity with causal CSIT is found by Shannon in \cite{shannon_CSIT}, and that with non-causal CSIT is found by Gelfand and Pinsker in \cite{gel1980coding}. This brings us to a key question: \emph{Can entanglement-assistance provide a quantum advantage for the capacity of a classical point-to-point channel with state?} Our goal in this work is to answer this question, specifically for the \emph{causal CSIT} setting.

\subsubsection*{Connection to Prior Work}Since classical channels can be regarded as special cases of quantum channels, it is immediately noteworthy that  the entanglement-assisted capacity of a quantum channel with non-causal CSIT was  studied by Dupuis \cite{dupuis_QCSIT}, and revisited by Pereg \cite{pereg_QCSIT} where the latter also investigated the case with causal CSIT. These works  obtain capacity expressions as non-convex optimization problems involving auxiliary quantum systems and auxiliary quantum channels. A challenging aspect of utilizing these capacity expressions, noted in \cite{dupuis_QCSIT}, is that despite having a `single-letter' expression, the capacity is not generally \emph{computable}  because there is no upper bound on the dimension of auxiliary quantum systems that need to be optimized for the capacity computation. To our knowledge there are no available specializations of these capacity results for \emph{classical} channels that would allow efficient computation of entanglement-assisted capacity for our setting. In fact, for our focused objective of determining the existence of a quantum advantage, general capacity optimizations will not be necessary in this work, and direct proofs will be provided.

\begin{figure}[htbp]
\center
\includegraphics[width=0.67\textwidth]{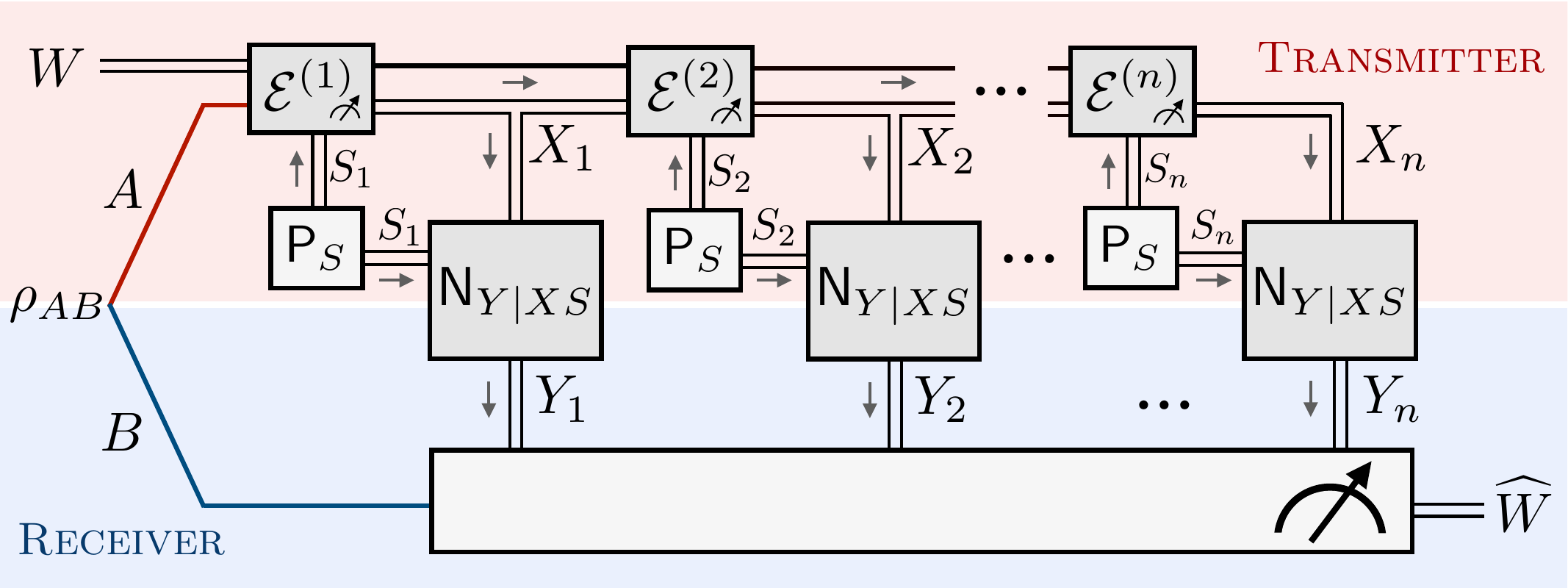}
\caption{General entanglement-assisted coding scheme with causal CSIT. A message $W$ is encoded into symbols $X_1,\dots,X_n$ that are sent by the transmitter over $n$ uses of a memoryless channel with state $\mN_{Y\mid XS}$. The state sequence $S_1,\dots,S_n\stackrel{\mbox{\tiny i.i.d.}}{\sim} \mP_S$ is revealed causally to the transmitter. The encoding of $X_i$ over the $i^{th}$ channel use is based on a quantum instrument depicted as $\mathcal{E}^{(i)}$ in the figure which can depend (the dependencies are not shown in the figure) on $W, S_1,\dots, S_i, X_1,\dots,X_{i-1}$, acting on the transmitter's side $(A)$ of an (entangled) quantum system $AB$ shared with the receiver in advance. The decoding at the receiver is based on a POVM that depends on the entire channel output sequence $Y_1,\dots,Y_n$, acting on $B$. See Section \ref{sec:EAcoding} for a formal description.}
\label{fig:QEAcoding}
\end{figure}

\subsubsection*{Overview of Results}
We formalize the framework of entanglement-assisted coding schemes for communication over a classical channel with \emph{causal} CSIT (illustrated in Fig. \ref{fig:QEAcoding}). Informally, a transmitter wants to (reliably) send a message to a receiver over a classical channel. The channel is discrete memoryless, and it depends on an i.i.d. random state, whose realization is revealed causally to the transmitter. The transmitter and the receiver may share additional (entangled) quantum resources that are generated and distributed before the communication begins. 
Given a channel with state, there are two metrics of particular interest in this work: 1) its capacity, associated with coding schemes that have vanishing error for large blocklengths, and 2) its zero-error capacity, associated with coding schemes that have zero error. Our main goal is to find channels with state whose classical capacity and/or zero-error capacity can be improved by entanglement assistance. The key results obtained in this work are the following.
\begin{enumerate}
  \item We identify a family of point-to-point channels with state (referred to as `\emph{graph channels with state},' see Definition \ref{def:graph_channel}), whose capacity with causal CSIT is improved by entanglement assistance. This is in sharp contrast to the point-to-point channel without state, where such an improvement is known to be impossible \cite{Bennett_Shor_Smolin_Thapliyal_PRL, matthews2012linear}. Remarkably, we find channels with state for which the capacity improvement from entanglement-assistance is quite significant, approaching $12/\pi^2\approx 21.6\%$. 
  \item We construct a channel with state and causal CSIT whose zero-error capacity is equal to $0$ without entanglement, but equal to $1$ with entanglement assistance. This is in contrast to the channel without state, for which it has been shown by Cubitt et al. in \cite{cubitt2010improving} that entanglement assistance \emph{cannot} activate (increase from zero to non-zero) the zero-error capacity.  It follows that the gain of zero-error capacity from entanglement assistance can be unbounded for a channel with causal CSIT.  Our construction of the channel with state is based on a bipartite Kochen Specker (B-KS) set \cite{BPQS} of orthogonal bases.  It is worthwhile to point out that \cite{BPQS} shows that a B-KS set can be made out of any bipartite quantum perfect strategy (including the one used in the Mermin-Peres magic square game \cite{mermin1990simple,peres1990incompatible}). Therefore, any bipartite quantum perfect strategy corresponds to such a channel with state. 
\end{enumerate}

\section{Problem Formulation}
\subsection{Notation}
Let $\mathbb{N}, \mathbb{Z}_{\geq 0}, \mathbb{R}_{\geq 0}$ denote the set of positive integers, the set of non-negative integers, and the set of non-negative reals, respectively. 
For $n\in \mathbb{N}$, we use $[n]$ to denote $\{1,2,\cdots, n\}$ and $x^n$ to denote $(x_1,x_2,\cdots,x_n)$. 

For random variables $X,Y,Z$, let $H(X), H(Y\mid X)$ denote the Shannon entropy, conditional entropy, respectively, and let $I(X;Y)$, $I(X;Y\mid Z)$ denote the mutual information and conditional mutual information, respectively. With a slight abuse of notation,  we also use $H(\mP)$ to denote the Shannon entropy of a random variable that has distribution $\mP$, i.e., $H(\mP) = -\sum_{i} \mP(i)\log_2 \mP(i)$, with the convention $0\log 0\triangleq 0$. $H_b(p) \triangleq H([p,1-p])$ denotes the binary entropy function.  $\mathbb{I}[x]$ denotes the indicator function, returning $1$ when the predicate $x$ is true and $0$ otherwise. $\Pr(E)$ denotes the probability of an event $E$. 
For compact notation, we may write $\Pr(A=a \mid B=b)$ as $\Pr(a\mid b)$, where uppercase letters denote random variables and the corresponding lowercase letters denote their realizations.

We use $\mathcal{H}_A$ to denote the Hilbert space associated with quantum system $A$. Let $\mathcal{L}(\mathcal{H})$ denote the set of linear operators on $\mathcal{H}$, and let $\mathcal{D}(\mathcal{H}) \subset \mathcal{L}(\mathcal{H})$ denote the set of density operators on $\mathcal{H}$. A quantum instrument $\mathcal{E}=\{\mathcal{E}_x\}_{x\in \mathcal{X}}$ on a Hilbert space $\mathcal{H}$ is a collection of completely-positive, trace-non-increasing maps $\mathcal{E}_x\colon \mathcal{L}(\mathcal{H}) \to \mathcal{L}(\mathcal{H})$ such that $\sum_{x\in \mathcal{X}}\mathcal{E}_x$ is trace preserving.  A positive operator-valued measure (POVM) $\{D_x\}_{x\in \mathcal{X}}$ on a Hilbert space $\mathcal{H}$ is a collection of positive semidefinite operators $D_x \in \mathcal{L}(\mathcal{H})$ such that $\sum_{x\in \mathcal{X}}D_x = I$ where $I$ is the identity operator on $\mathcal{H}$. We use ${\rm id}(\cdot)$ to denote the identity channel. Let $\Tr[\cdot]$ denote the trace of an operator and $\Tr_A[\cdot]$ denote the partial trace over system $A$. We use $M^\top, M^*$ and $M^\dagger$ to denote the transpose, complex conjugate, and conjugate transpose of a matrix $M$, respectively.

\subsection{Communication over classical channel with causal CSIT}
We adopt the standard model of a classical channel with state (cf. \cite[Sec. 7]{NIT}). A channel with state is specified by $(\mN, \mP_S)$ with underlying finite alphabets $(\mathcal{X}, \mathcal{Y}, \mathcal{S})$. $\mN (y\mid x,s)$ specifies the channel's conditional probability distribution, i.e., the probability of output $Y=y$ given the input $X=x$ and the channel state $S=s$, for $x\in \mathcal{X}, s\in \mathcal{S}, y\in \mathcal{Y}$. $\mP_S(s)$ specifies the probability of the channel state being $S=s$ for $s\in \mathcal{S}$. Without loss of generality, we assume $\mP_S(s)>0$ for all $s\in \mathcal{S}$.

We assume that the channel is memoryless and the state is i.i.d. across channel uses. Specifically, for $n$ uses of the channel, $X^n, S^n, Y^n$ collectively denote the inputs, states, and the outputs corresponding to the $n$ channel uses, respectively. The probability distribution of $S^n$ is $\mP_{S^n}(s^n) \triangleq \prod_{i=1}^n \mP_S(s_i) \triangleq \mP_S^{\otimes n}(s^n)$. The channel's conditional distribution for $n$ uses of the channel is $\prod_{i=1}^n \mN (y_i\mid x_i,s_i) \triangleq \mN^{\otimes n}(y^n\mid x^n, s^n)$, for $s^n\in \mathcal{S}^n$, $x^n\in \mathcal{X}^n$, $y^n \in \mathcal{Y}^n$.
 
A message $W\in [M]$ originates uniformly at the transmitter, is encoded into a sequence of $n$ symbols (a codeword) $X_1,X_2,\dots, X_n$ that are input into the channel over $n$ channel uses. We mainly focus on the causal CSIT case.  Causal CSIT implies that each codeword symbol can depend on the current and past channel states but not the future channel states. Such a scheme is called an $(M,n)$ coding scheme, where $M$ is the size of the message and $n$ is the blocklength.  Depending on the type of resources that the encoder and decoder are allowed to share in advance, the following two progressively more capable scenarios are explored in this work: 1) they can share classical randomness, 2) they can share entangled quantum resources. Each setting is formalized by defining its feasible space of coding schemes.

\begin{remark} \label{rem:input_constraint}
  In this paper we sometimes find it convenient to model the state of the channel as imposing input constraints on a channel without state. Let ${\mN}_o(y\mid x)$ be a channel without state with input alphabet $\mathcal{X}$ and output alphabet $\mathcal{Y}$, and let $\mathcal{X}_s\subseteq \mathcal{X}$ for $s\in \mathcal{S}$. When the channel state is $S=s$, the transmitter is restricted to choosing an input $x\in \mathcal{X}_s$. Such a channel can be viewed as a standard channel with state by defining $\mN(y\mid x,s) = \mN_o(y\mid x)$ if $x\in \mathcal{X}_s$ and $\mN(y\mid x,s) = \mN_o(y\mid x_s)$ if $x\not\in \mathcal{X}_s$, for some reference symbol $x_s\in \mathcal{X}_s$.
\end{remark}

\subsection{Classical coding schemes with causal CSIT}
A classical $(M,n)$ coding scheme with causal CSIT is specified by a tuple $(\mP_{Q}, \{\phi^{(i)}\}, \psi)$. Both the transmitter and the receiver may additionally have access to a random variable $Q\in \mathcal{Q}$, which is generated independently of the message $W$. $\mP_Q$ denotes the distribution of $Q$. For $i\in [n]$, $\phi^{(i)} \colon [M]\times \mathcal{S}^{i} \times \mathcal{Q} \to \mathcal{X}$ specifies the encoder at the $i^{th}$ channel use, such that $X_i = \phi^{(i)}(W,S^i,Q)$, which depends on the message, the current and the past channel states and the random variable $Q$. At the receiver, $\psi\colon \mathcal{Y}^n\times \mathcal{Q} \to [M]$ specifies the decoder, such that the decoded message $\widehat{W} = \psi(Y^n,Q)$. The probability of success (decoding the message correctly) is,
\begin{equation}
\begin{aligned}
  &\Pr(\widehat{W} =W) = \frac{1}{M} \sum_{q} \mP_Q(q)  \sum_{w} \sum_{s^n} \mP_{S}^{\otimes n}(s^n)   \times \\
  &\qquad\qquad\qquad \sum_{y^n\colon \psi(y^n,q)=w}~ \prod_{i=1}^n \mN\big(y_i\mid \phi^{(i)}(w,s^i,q),s_i\big).
   \label{eq:prob_succ_classical}
\end{aligned}
\end{equation}
\begin{remark} \label{rem:det_coding}
  The optimal probability of success, over the set of classical $(M,n)$ coding schemes, is achieved by $|\mathcal{Q}|=1$, i.e., a deterministic scheme without (common) randomness, by noting that $\exists q^*\in \mathcal{Q}$ such that Eq. $\eqref{eq:prob_succ_classical}= \sum_q \mP_Q(q) f(q) \leq f(q^*)$, implying the existence of an optimal scheme that needs no randomness.
\end{remark}

\subsection{Entanglement-assisted coding schemes with causal CSIT}\label{sec:EAcoding}
An entanglement-assisted $(M,n)$ coding scheme with causal CSIT allows the transmitter and the receiver to share in advance any entangled quantum resource, which is used to assist both the encoding and decoding processes.   
Formally, such a scheme is specified by a tuple $$(\rho_{AB}, \{\mathcal{E}^{(i)}\}_{i\in [n]}, \{D_{\widehat{w}\mid y^n}\}_{y^n\in \mathcal{Y}^n, \widehat{w}\in [M]})$$. 
Before the communication, the transmitter and the receiver share a bipartite quantum system $AB$ in the state $\rho_{AB} \in \mathcal{D}(\mathcal{H}_A \otimes \mathcal{H}_B)$, such that $A$ is with the transmitter and $B$ is with the receiver. 

At the first time slot, conditioned on $W=w$ and $S_1=s_1$, the transmitter applies the (encoding) quantum instrument $\mathcal{E}^{(1)} = \{\mathcal{E}_{x_1\mid (w,s_1)}^{(1)}\}_{x_1\in \mathcal{X}}$ to $A$ and observes the classical outcome $X_1$, which is then sent to the channel. 
For $x_1\in \mathcal{X}$, the outcome $X_1=x_1$ happens with conditional probability $\Pr(x_1\mid w, s_1)=\Tr\big[\rho_{x_1\mid w,s_1}\big]$, where $\rho_{x_1\mid w,s_1} \triangleq (\mathcal{E}_{x_1\mid (w,s_1)}^{(1)} \otimes {\rm id}_B)(\rho_{AB})$ denotes the  unnormalized post-measurement state. We denote the normalized post-measurement state as $\rho_{AB\mid w,s_1,x_1} = \rho_{x_1\mid w,s_1}/\Pr(x_1\mid w, s_1)$.

Sequentially for $i=2,3,\cdots, n$, in the $i^{th}$ channel use, conditioned on $(W=w,S^i=s^i,X^{i-1}=x^{i-1})$, the transmitter applies the $i^{th}$ quantum instrument $\mathcal{E}^{(i)} = \{\mathcal{E}_{x_i \mid (w,s^i,x^{i-1})}^{(i)}\}_{x_i\in \mathcal{X}}$ to $A$ and observes the classical outcome $X_i$, which is then sent to the channel. 
For $x^i\in \mathcal{X}^i$, the (joint) outcome $X^i=x^i$ happens with conditional probability $\Pr(x^i\mid w,s^i)= \Tr\big[\rho_{x^i\mid w,s^{i}}\big]$, where $\rho_{x^i\mid w,s^{i}} \triangleq (\mathcal{E}_{x_i \mid (w,s^i,x^{i-1})}^{(i)} \otimes {\rm id}_B)(\rho_{x^{i-1}\mid w,s^{i-1}})$ denotes the unnormalized post-measurement state. We denote the normalized post-measurement state as $\rho_{AB\mid w,s^i,x^i} = \rho_{x^i\mid w,s^i}/\Pr(x^i\mid w,s^i)$.

After $n$ channel uses, the receiver collects $Y^n$ from the channel. Conditioned on $Y^n=y^n$, the receiver applies the POVM $\{D_{\widehat{w}\mid y^n}\}_{\widehat{w} \in [M]}$ on $B$ and observes the classical outcome $\widehat{W}$ as the decoded message. 
The conditional probability $\Pr(\widehat{w} \mid w, s^n, x^n, y^n) = \Tr\big[ (I\otimes D_{\widehat{w}\mid y^n}) \rho_{AB\mid w,s^n,x^n} \big]$. The probability of success is,
\begin{align}
	&\Pr(\widehat{W}=W) = \frac{1}{M}\sum_{w}  \sum_{s^n,x^n, y^n} \mP_S^{\otimes n}(s^n) \times   \mN^{\otimes n}(y^n\mid s^n,x^n) \notag \\
	&\qquad \qquad\qquad \times   \Tr\big[(I\otimes D_{w\mid y^n}) \rho_{x^n\mid w,s^n} \big]
\end{align}
Since classical common randomness can be modeled as a part of the shared quantum resource, it is clear that the entanglement-assisted coding schemes are at least as capable as  classical coding schemes.
 
\subsection{Rate and capacity}
A (communication) rate $R\in \mathbb{R}_{\geq 0}$ is said to be  achievable classically (resp. with entanglement assistance) if (and only if) there exists a sequence (indexed by $n$) of classical (resp. entanglement-assisted) $(M_n,n)$ coding schemes (with the message and the decoded message denoted by $W_n$ and $\widehat{W}_n$, respectively), such that   
\begin{align}
  \lim_{n\to \infty} \Pr(\widehat{W}_n=W_n) = 1 ~~\mbox{and} ~~\lim_{n\to \infty} \frac{\log_2 M_n}{n} \geq R.
\end{align}
For a channel with state $(\mN, \mP_S)$, let $C^{\C}$ denote its classical capacity with causal CSIT, defined as the supremum of achievable rates using classical coding schemes. The entanglement-assisted capacity with causal CSIT, denoted by $C^{\EA}$, is defined similarly for entanglement-assisted coding schemes.

\subsection{Zero-error capacity}
For a channel with state $(\mN, \mP_S)$ and $n\in\mathbb{N}$, let $M_{n,\opt}^{\C}$ (resp. $M_{n,\opt}^{\EA}$) be the maximum value of $M$, such that there exists a classical (resp. entanglement-assisted) $(M,n)$  coding scheme with probability of success $\Pr(\widehat{W}=W)=1$, i.e., probability of error is $0$.
We define $c_0^{\C} \triangleq \log_2(M_{1,\opt}^{\C})$ as the one-shot zero-error capacity, and similarly $c_0^{\EA} \triangleq \log_2(M_{1,\opt}^{\EA})$ as the entanglement-assisted one-shot zero-error capacity. 
We further define the zero-error capacity (with causal CSIT) as $C_0^{\C} \triangleq \sup \big\{ \frac{1}{n}\log_2 (M_{n,\opt}^{\C}) \colon n\in \mathbb{N} \big\}$. Similarly, the entanglement-assisted zero-error capacity (with causal CSIT) is defined as $C_0^{\EA} \triangleq \sup\big\{ \frac{1}{n}\log_2 (M_{n,\opt}^{\EA}) \colon n\in \mathbb{N} \big\}$.

\section{Results} \label{sec:results}
Recall that we consider only classical channels with causal CSIT throughout this work. The notations associated with the various notions of capacity are summarized in Table \ref{tab:capacity}.
\begin{table}[htbp]
\caption{Notation for different notions of capacity}
\label{tab:capacity}
\centering
\begin{tabular}{|c|c|}
	\hline Notation & Description \\ \hline 
	$C^{\C}$ & $\Pr({\rm error}) \to 0$, without entanglement assistance\\\hline
	$C^{\EA}$ & $\Pr({\rm error}) \to 0$, with entanglement assistance \\\hline
	$C^{\C}_0$ & $\Pr({\rm error}) = 0$, without entanglement assistance \\\hline
	$C^{\EA}_0$ & $\Pr({\rm error}) = 0$, with entanglement assistance \\\hline
	$c^{\C}_0$ & One-shot, $\Pr({\rm error}) = 0$, without entanglement assistance\\\hline
	$c^{\EA}_0$ & One-shot, $\Pr({\rm error}) = 0$, with entanglement assistance \\\hline
\end{tabular}
\end{table}
\subsection{Classical capacity with causal CSIT can be improved by entanglement assistance}
We first show that with \emph{causal} CSIT,  for certain channels with state, the classical capacity is improved by quantum entanglement assistance, $C^{\EA} > C^{\C}$. Our examples include the following class of channels labeled the \emph{graph channel with state} (Definition \ref{def:graph_channel}), along with their noisy versions (Definition \ref{def:noisy}). 
 
\begin{definition}[Graph channel with state] \label{def:graph_channel}
A graph channel with state is specified by $(\mathcal{G},\mP_S)$, where $\mP_S$ is the state distribution and $\mathcal{G}(\mathcal{Y},\mathcal{S})$ is a simple undirected graph whose vertex set $\mathcal{Y}$ corresponds to the channel output alphabet and whose edge set $\mathcal{S}$ corresponds to the channel state alphabet. 
The channel input alphabet is $\mathcal{X}=\{0,1\}$.  
Conditioned on the channel state $S=s$, the transmitter selects one of the two endpoints of the edge $s$ as the channel output observed by the receiver. More precisely, for each edge $s\in \mathcal{S}$, let $s_0,s_1\in \mathcal{Y}$ identify\footnote{The order of $s_0$ and $s_1$ does not essentially change the problem; however they must be fixed for a graph channel with state. For example, without loss of generality we may assume $\mathcal{Y}\subset\mathbb{Z}_{\geq 0}$ and $s_0<s_1$ to fix an order.} the two vertices incident to $s$. Then
\begin{align}
  \mN(y\mid x,s) = \mathbb{I}[y = s_x], 
\end{align}
for all $x\in \mathcal{X}, s\in \mathcal{S}, y\in \mathcal{Y}$.
\end{definition}
\begin{remark}
  Although our definition of a graph channel with state is expressed based on a graph, it should not be confused with the confusability graphs for a channel (without state) in zero-error channel capacity (e.g. \cite{Shannon56}), where two inputs (vertices) are connected by an edge if and only if they can (with positive probability) yield the same output.
\end{remark}

For an integer $m\geq 3$, let $\mathcal{C}_m$ denote the \emph{cyclic} graph with $m$ vertices, and $\mathcal{K}_m$ denote the \emph{complete} graph with $m$ vertices. Fig. \ref{fig:example_graph} illustrates two examples of such graph channels with state.

We write $\unif$ to represent the uniform distribution for the channel state, i.e., $\unif(s) = |\mathcal{S}|^{-1}$ for all $s$. We are now ready to present our first theorem.
\begin{theorem}[Cyclic graph $\mathcal{C}_5$]\label{thm:cyclic_graph}
  For the graph channel with state specified by $(\mathcal{C}_5,\unif)$, denote its classical capacity with causal CSIT as $C^{\C}(\mathcal{C}_5)$, and its entanglement-assisted capacity with causal CSIT as $C^{\EA}(\mathcal{C}_5)$. We have
  \begin{align}
    C^{\C}(\mathcal{C}_5) = 0.8,
  \end{align}
  and 
  \begin{align} \label{eq:QE_C5}
    C^{\EA}(\mathcal{C}_5) \geq 1-H_b\big(\cos^2(\pi/20)\big) \approx 0.8341.
  \end{align}
\end{theorem}
\noindent Theorem \ref{thm:cyclic_graph} provides an example of a channel with state for which the entanglement-assisted capacity surpasses its classical capacity. Notably, the channel is deterministic, meaning that the output $Y$ is determined by $(X,S)$.
 The proof of Theorem \ref{thm:cyclic_graph} is provided in Section  \ref{proof:C5}.  As a brief overview, the entanglement-assisted scheme achieving the rate in the RHS of \eqref{eq:QE_C5} works as follows. For each channel use, the transmitter and the receiver consume a Bell pair to convert the channel to a binary symmetric channel (BSC) $X'\to Y'$ with crossover probability $\Pr(Y'\neq X')= 1-\cos^2(\pi/20)$. Notably, the channel conversion process over the $i^{th}$ channel use depends on the the current state $s_i$, and $X_i'$, which is the current input of the converted channel, but does not depend on the past states $s^{i-1}$. See  Fig. \ref{fig:conversion} for an illustration of the channel conversion process. Since the capacity of a BSC with crossover  probability $1-p$ is $1-H_b(1-p) = 1-H_b(p)$, a classical BSC code can  be employed over the converted channel to achieve the rate on the RHS of \eqref{eq:QE_C5}.

\begin{figure}[htbp]
\center
\includegraphics[width=0.67\textwidth]{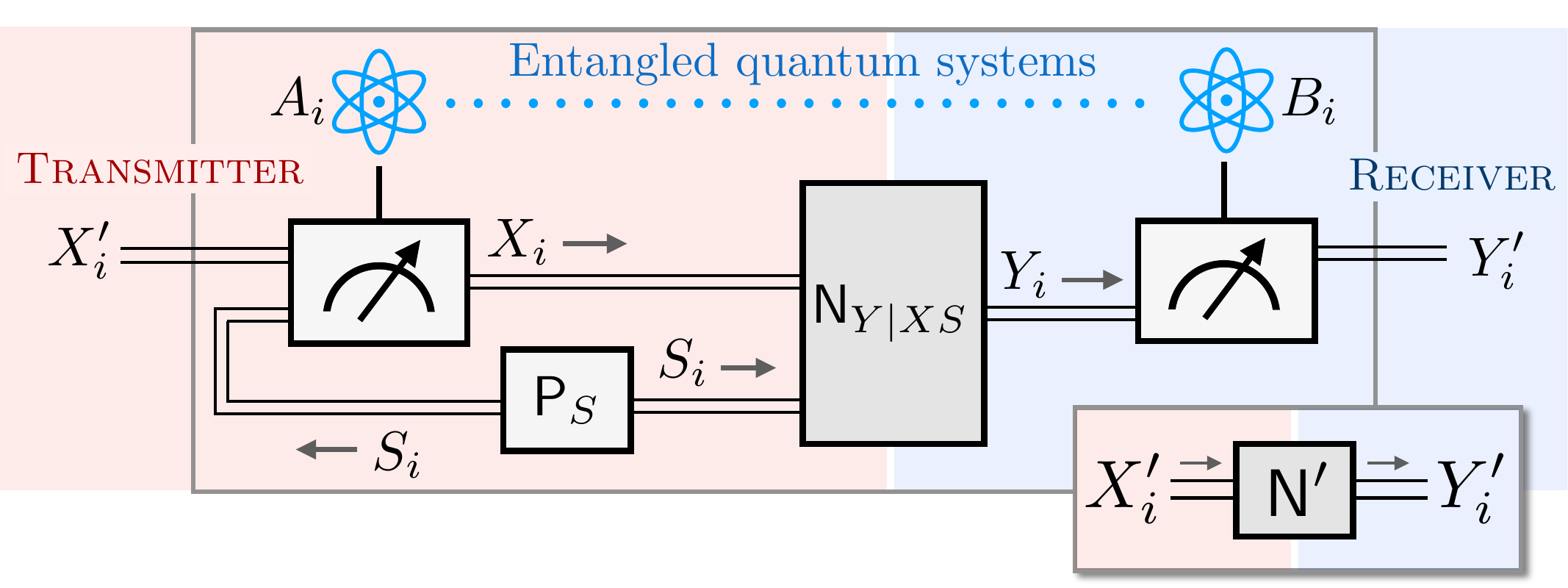}
\caption{A  strategy that converts each use (say, the $i^{th}$ use) of a channel $\mN$ with state  ($S_i$) and entanglement assistance ($A_iB_i$), informally $(A_i,S_i,X_i)\stackrel{\mN}{\longrightarrow} (B_i,Y_i)$, into an effective classical channel $\mN'$ with input $X_i'$ and output  $Y_i'$, that has no state and no entanglement assistance, informally $X_i'\stackrel{\mN'}{\longrightarrow} Y_i'$.}
\label{fig:conversion}
\end{figure}

To further generalize the findings, our next theorem identifies a family of channels with state where the gain (factor) from entanglement assistance  is even larger. We need the following definition.
\begin{definition}[Noisy version] \label{def:noisy}
  Given $p \in [0,1]$, and a channel with state $(\mN,\mP_S)$, the noisy version of the channel with state (depending on a parameter $p$), denoted as $(\mN_p, \mP_S)$, is defined by its conditional probability distribution
  \begin{align}
    \mN_p(y\mid x,s) = p\cdot \mN(y\mid x,s) + (1-p) \cdot |\mathcal{Y}|^{-1},
  \end{align}
  for all $x\in \mathcal{X}, s\in \mathcal{S}, y\in \mathcal{Y}$.
  Equivalently, $\mN_p$ can be thought of as replacing  the output of the channel $\mN$ with probability $1-p$ by a random noise uniformly distributed over $\mathcal{Y}$.
\end{definition}

\begin{theorem}[Noisy complete graph $\mathcal{K}_m$] \label{thm:complete_graph}
  For the noisy version of the graph channel with state $(\mathcal{K}_m,\unif)$, denote its classical capacity with causal CSIT as $C^{\C}(\mathcal{K}_m, p)$, and its entanglement-assisted capacity with causal CSIT as $C^{\EA}(\mathcal{K}_m, p)$. We have
  \begin{align} \label{eq:CC_Km}
    C^{\C}(\mathcal{K}_m, p) = \frac{1}{m}\sum_{k=1}^m \big(1+  pc_{m,k}  \big)\log_2 \big( 1+  pc_{m,k}  \big),
  \end{align}
  where $c_{m,k} \triangleq  \tfrac{m-2k+1}{m-1}$,
  and
  \begin{align} \label{eq:RQ_Km}
    C^{\EA}(\mathcal{K}_m,p) &\geq \max\big\{C^{\C}(\mathcal{K}_m,p),~1-H_b(q_{m,p}) \big\} \notag \\
    &\triangleq R^{\EA}(\mathcal{K}_m,p),
  \end{align}
  where $q_{m,p}\triangleq \big( \tfrac{\cot(\pi/2m)}{2(m-1)}+\tfrac{1}{2}\big)p + \tfrac{1}{2}(1-p)$.
\end{theorem}
\noindent The proof of Theorem \ref{thm:complete_graph} is presented in Appendix \ref{proof:complete_graph}. Cases associated with $\mathcal{K}_8$ are illustrated in Fig. \ref{fig:rate_comparison}.

\begin{corollary} \label{cor:gain_factor}
  As $p\to 0^+$, i.e., when the channel is almost completely noisy, the limit
  \begin{align} \label{eq:ratio_complete_graph}
    \lim_{p \to 0^+} \frac{R^{\EA}(\mathcal{K}_m,p)}{C^{\C}(\mathcal{K}_m, p)} = \frac{3\cot^2\big( \tfrac{\pi}{2m} \big)}{m^2-1},
  \end{align}
  and for $m\to \infty$ the ratio in \eqref{eq:ratio_complete_graph} approaches  $\tfrac{12}{\pi^2}> 1.2158$.
\end{corollary}
\noindent The proof of Corollary \ref{cor:gain_factor} is presented in Appendix \ref{proof:gain}. Evidently, there is a channel with state for which entanglement assistance can improve upon the classical capacity by at least $21.58\%$. 

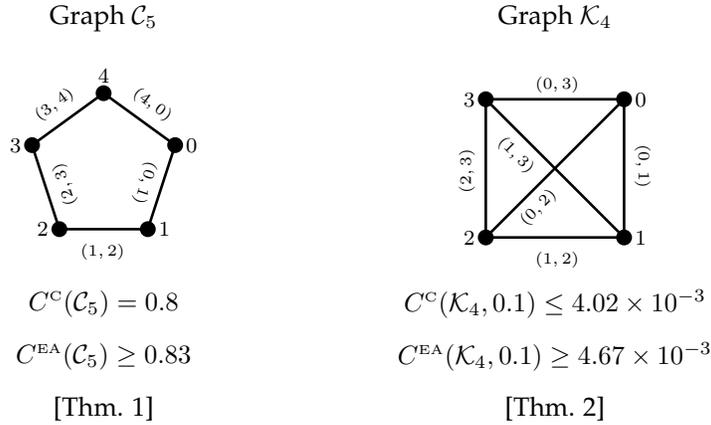
\begin{figure}[h]
\centering
\begin{tikzpicture}[scale=1, line width=1pt]
\begin{scope}
\foreach \i in {0,...,4} {
    \coordinate (v\i) at ({72*\i+18}:1);
  }
  \draw (v0)--(v1)--(v2)--(v3)--(v4)--cycle;
  \foreach \i in {0,...,4} {
    \fill (v\i) circle (3pt);
  }
  \node[font=\scriptsize, right] at (v0) {$0$};
  \node[font=\scriptsize, above] at (v1) {$4$};
  \node[font=\scriptsize, left] at (v2) {$3$};
  \node[font=\scriptsize, left] at (v3) {$2$};
  \node[font=\scriptsize, right] at (v4) {$1$};
  \node[font=\tiny, rotate=36] at ($(v1.west)+(-0.65,-0.15)$) {$(3,4)$};
  \node[font=\tiny, rotate=0] at ($(v4.east)+(-0.6,-0.3)$) {$(1,2)$};
  \node[font=\tiny, rotate=-36] at ($(v1.east)+(0.65,-0.15)$) {$(4,0)$};  
  \node[font=\tiny, rotate=-108] at ($(v2.east)+(1.5,-0.5)$) {$(0,1)$}; 
  \node[font=\tiny, rotate=108] at ($(v4.east)+(-1.1,0.6)$) {$(2,3)$};
  \node[font=\small, rotate=0] at (0,2) {Graph $\mathcal{C}_5$};
  
  \node at (0,-1.8) {\small $C^{\C}(\mathcal{C}_5) = 0.8$};
  \node at (0,-2.5) {\small $C^{\EA}(\mathcal{C}_5) \geq 0.83$};
  \node at (0,-3.2) {\small [Thm. \ref{thm:cyclic_graph}]};
\end{scope}

\begin{scope}[shift={(6,0)}]
  \foreach \i in {0,...,3} {
    \coordinate (v\i) at ({90*\i + 45}:1.3);
    \fill (v\i) circle (3pt);
  }
  \draw (v0) -- (v1);
  \draw (v1) -- (v2);
  \draw (v2) -- (v3);
  \draw (v3) -- (v0);
  \draw (v0) -- (v2);
  \draw (v1) -- (v3);
  \node[font=\scriptsize, right] at (v0) {$0$};
  \node[font=\scriptsize, left] at (v1) {$3$};
  \node[font=\scriptsize, left] at (v2) {$2$};
  \node[font=\scriptsize, right] at (v3) {$1$};
  \node[font=\tiny, rotate=0] at ($(v1.east)+(0.95,0.2)$) {$(0,3)$};
  \node[font=\tiny, rotate=0] at ($(v2.east)+(0.95,-0.3)$) {$(1,2)$};
  \node[font=\tiny, rotate=90] at ($(v1.west)+(-0.25,-0.95)$) {$(2,3)$};
  \node[font=\tiny, rotate=-90] at ($(v3.east)+(0.25,0.95)$) {$(0,1)$};
  \node[font=\tiny, rotate=-45] at ($(v2.east)+(0.4,1.1)$) {$(1,3)$};
  \node[font=\tiny, rotate=45] at ($(v2.east)+(0.7,0.4)$) {$(0,2)$};
  \node[font=\small, rotate=0] at (0,2) {Graph $\mathcal{K}_4$};
  
  \node at (0,-1.8) {\small $C^{\C}(\mathcal{K}_4,0.1) \leq 4.02\times 10^{-3}$};
  \node at (0,-2.5) {\small $C^{\EA}(\mathcal{K}_4,0.1) \geq 4.67\times 10^{-3}$};
  \node at (0,-3.2) {\small [Thm. \ref{thm:complete_graph}]};
\end{scope}
\end{tikzpicture}
\caption{Two examples of graph channels with state. The left-hand side shows $\mathcal{C}_5$ and the right-hand side shows $\mathcal{K}_4$. For each graph channel with state, each vertex (corresponding to an output symbol) is labeled by a value $y\in \mathcal{Y}$ (without loss of generality we let $\mathcal{Y} = \{0,1,\cdots, m-1\}$), and each edge (corresponding to a channel state) is labeled by a tuple $s=(s_0,s_1)$, where $s_0, s_1\in \mathcal{Y}$ denote the two endpoints of the edge $s$. Given the channel state $S=s$ and input $X\in \{0,1\}$, the receiver observes the output $Y=s_X$.}
\label{fig:example_graph}
\end{figure}

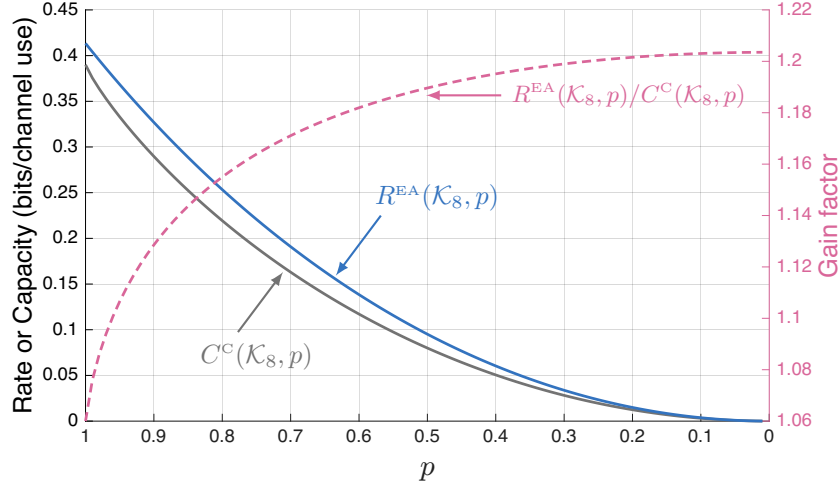
\begin{figure}
\centering
\begin{tikzpicture}
\definecolor{color1}{rgb}{0.45,0.45,0.45}
\definecolor{color2}{rgb}{0.20,0.45,0.75}
\definecolor{color3}{rgb}{0.85,0.40,0.60}

\node (myfirstpic) at (0,0) {\includegraphics[width=0.67\textwidth]{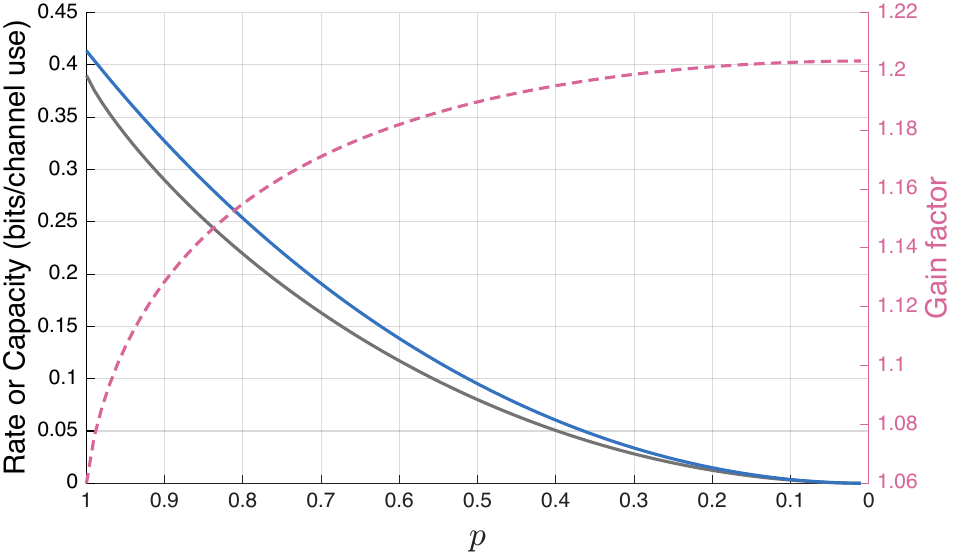}};
\draw[{latex}-, thick, color1] (-1.9,-0.4) -- (-2.5,-1.2)  node [right=-0.4, pos=1.4] {\small $C^{\C}(\mathcal{K}_8,p)$};
\draw[{latex}-, thick, color2] (-1.2,-0.45) -- (-0.6,0.35)  node [right=-0.4, pos=1.3] {\small $R^{\EA}(\mathcal{K}_8,p)$};
\draw[{latex}-, thick, color3] (0,1.95) -- (1,1.95)  node [right] {\footnotesize $R^{\EA}(\mathcal{K}_8,p)/C^{\C}(\mathcal{K}_8,p)$};

\end{tikzpicture}
\caption{Plots of the classical capacity $C^{\C}(\mathcal{K}_8,p)$ (Eq. \eqref{eq:CC_Km}), achievable rate by entanglement assistance $R^{\EA}(\mathcal{K}_8,p)$ (Eq. \eqref{eq:RQ_Km}), and the gain factor $R^{\EA}(\mathcal{K}_8,p)/C^{\C}(\mathcal{K}_8,p)$, for the graph channel with state associated with $\mathcal{K}_8$, as functions of $p$.}
\label{fig:rate_comparison}
\end{figure}

\subsection{Zero-error capacity with causal CSIT can be activated by entanglement assistance}
Next, we shift our focus from the classical capacity to the zero-error capacity. The entanglement-assisted zero-error capacity for a channel without state is studied by \cite{cubitt2010improving}, and it is shown that if the one-shot zero-error capacity is $0$, then it remains at $0$ even with entanglement assistance. In other words, the one-shot zero-error capacity of a classical channel (without state) cannot be \emph{activated} by entanglement assistance.  A natural question is then to ask whether the statement still holds in the setting of channels with state. To address this question, we first look at the graph channels with state, and we have the following theorem.
\begin{theorem} \label{thm:no_activation_for_graph}
  Let $c_0^{\C}(\mathcal{G})$ (resp. $c_0^{\EA}(\mathcal{G})$) denote the classical (resp. entanglement-assisted) one-shot zero-error capacity of a graph channel with state defined by $\mathcal{G}$. Then, $c_0^{\C}(\mathcal{G})= 0 \Leftrightarrow C_0^{\C}(\mathcal{G})= 0 \Leftrightarrow$ ``$\mathcal{G}$  ~is~not~bipartite\footnote{A graph $\mathcal{G}(\mathcal{Y}, \mathcal{S})$ is bipartite if its vertex set $\mathcal{Y}$ can be partitioned into two disjoint sets $\mathcal{Y}_1$ and $\mathcal{Y}_2$ such that every edge in $\mathcal{S}$ has one endpoint in $\mathcal{Y}_1$ and the other in $\mathcal{V}_2$.}" $\Leftrightarrow c_0^{\EA}(\mathcal{G}) = 0$. 
\end{theorem}
\noindent The proof of Theorem \ref{thm:no_activation_for_graph} is presented in Appendix \ref{proof:no_activation_graph}. It follows from Theorem \ref{thm:no_activation_for_graph}  that for any \emph{graph channel with state}, the one-shot zero-error capacity cannot be activated by entanglement assistance. We do not know whether the conditions also imply $C^{\EA}(\mathcal{G}) = 0$, particularly because our proof of $c^{\EA}(\mathcal{G}) = 0$ makes use of the structure of graph channels with state. The argument cannot be extended to multiple uses of a channel with state by viewing it as a bigger channel, particularly because the bigger channel no longer falls into the class of graph channels with state.

However, surprisingly, there do exist other classical channels with state whose one-shot zero-error capacity and zero-error capacity (with causal CSIT) can be activated by entanglement assistance. To show such an example, we need the following definition of a bipartite Kochen-Specker (B-KS) set. 
\begin{definition}[B-KS set \cite{BPQS}]
  A B-KS set  in $\mathbb{C}^d$ is specified by a tuple $(\{\mathcal{A}_i\}_{i=1}^a,\{\mathcal{B}_{j}\}_{j=1}^{b})$, where for each $i\in [a]$, $\mathcal{A}_i$ is a complete orthogonal basis of $\mathbb{C}^d$, and for each $j\in [b]$,  $\mathcal{B}_j$ is a complete orthogonal basis of $\mathbb{C}^d$. The bases satisfy the property that given any selections of vectors $u_1\in \mathcal{A}_1, u_2\in \mathcal{A}_2,\cdots, u_a\in \mathcal{A}_a, v_1\in \mathcal{B}_1, v_2\in \mathcal{B}_2,\cdots, v_b\in \mathcal{B}_b$, there exist $i\in[a]$ and $j\in [b]$ such that $u_i$ and $v_{j}$ are orthogonal.
\end{definition}

Reference \cite{BPQS} provides a general method to construct a B-KS set given a bipartite perfect quantum strategy of a non-local game. For example, the magic square game  \cite{mermin1990simple,peres1990incompatible},  yields the following B-KS set \cite[Eq. (9a) -- (10c)]{BPQS}.
\begin{equation} \label{eq:BKS}
\begin{aligned}
  \mathcal{A}_1 &= \{(1,0,0,0), (0,1,0,0),(0,0,1,0),(0,0,0,1)\}\\
  \mathcal{A}_2 &= \{(1,1,1,1),(1,-1,1,-1),(1,1,-1,-1),(1,-1,-1,1)\}\\
  \mathcal{A}_3 &= \{(1,1,1,-1),(1,1,-1,1),(1,-1,1,1),(-1,1,1,1)\}\\
  \mathcal{B}_1 &= \{(1,1,0,0),(1,-1,0,0),(0,0,1,1),(0,0,1,-1)\}\\
  \mathcal{B}_2 &= \{(1,0,1,0),(0,1,0,1),(1,0,-1,0),(0,1,0,-1)\}\\
  \mathcal{B}_3 &= \{(1,0,0,1),(1,0,0,-1),(0,1,1,0),(0,1,-1,0)\}
\end{aligned}
\end{equation}

\begin{theorem}[Activation] \label{thm:activation}
  For any B-KS set in $\mathbb{C}^d$, $(\{\mathcal{A}_i\}_{i=1}^a,\{\mathcal{B}_{j}\}_{j=1}^{b})$, one can construct a classical channel with state $\mN$ with,
  \begin{align}
    &c_0^{\C}(\mN) = C_0^{\C}(\mN) = 0,\\ 
    & c_0^{\EA}(\mN) = C_0^{\EA}(\mN) = 1.
  \end{align}
\end{theorem}
\noindent The proof of Theorem \ref{thm:activation} is presented in Section \ref{proof:activation}. For a brief illustration of how the channel with state is constructed, in particular for the B-KS set specified in \eqref{eq:BKS}, see Fig. \ref{fig:BKS} and its description. We point out that the construction of the channel with state builds upon the construction in \cite[Thm. 2]{cubitt2010improving} of a channel (without state). A key difference is that here we additionally impose state-dependent input constraints.

\begin{figure}[htbp]
\center
\includegraphics[width=0.45\textwidth]{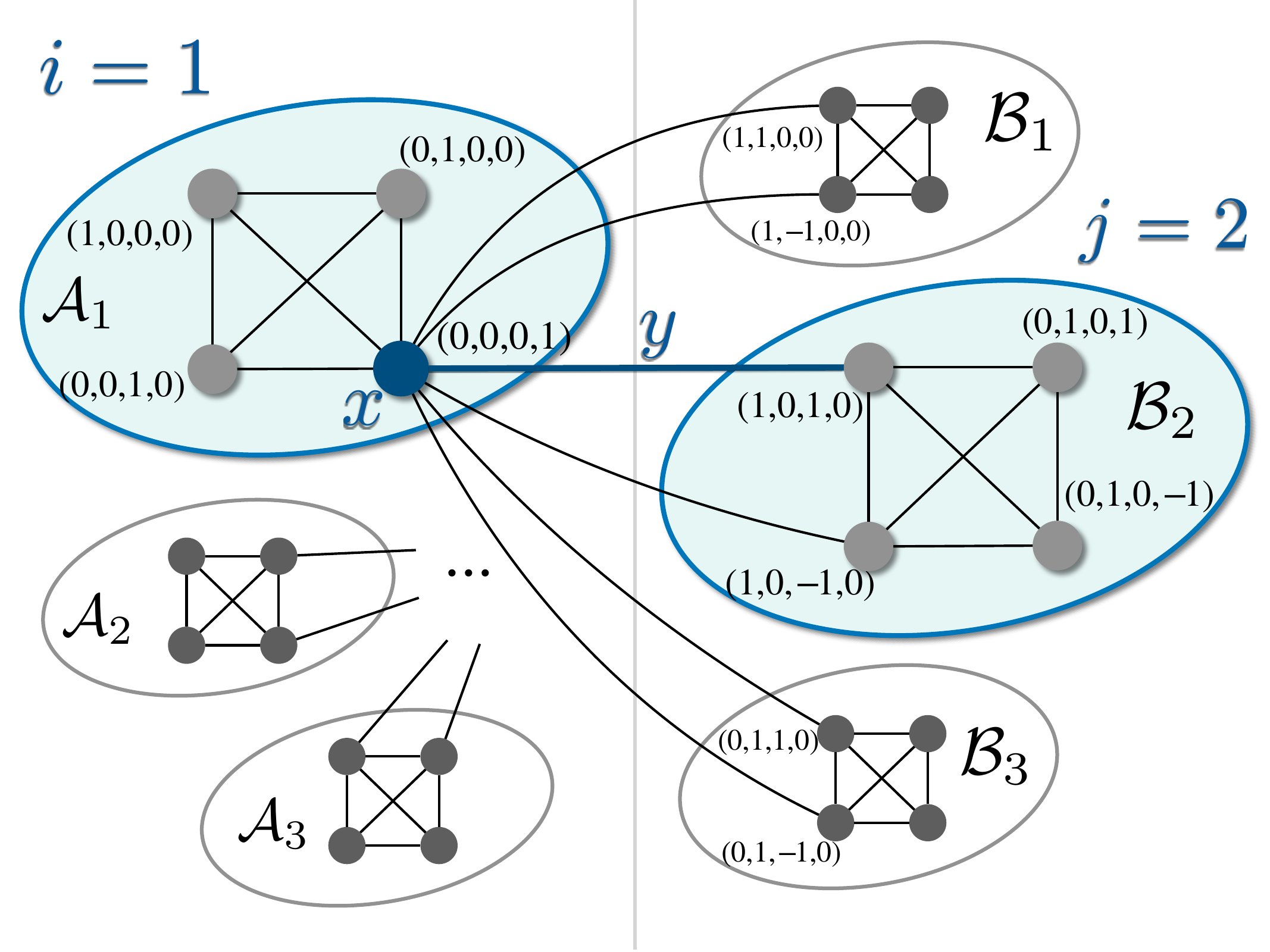}
\caption{A channel with state defined from the B-KS set \eqref{eq:BKS}: For each $\mathcal{A}_i, i\in \{1,2,3\}$ (and similarly $\mathcal{B}_j, j\in \{1,2,3\}$), there are $4$ channel inputs (dark circles in the figure) associated with the $4$ vectors in $\mathcal{A}_i$ ($\mathcal{B}_j$). The channel thus contains $24$ inputs in total.
The channel state $S\in \{(i,j)\colon i\in \{1,2,3\}, j\in \{1,2,3\}\}$, such that when the state is $S=(i,j)$, the transmitter is restricted to picking an input from the subsets of inputs associated with $\mathcal{A}_i \cup \mathcal{B}_j$ (See Remark \ref{rem:input_constraint}). Two inputs are connected by an edge if and only if their associated vectors are orthogonal (not all such edges are drawn in the figure above). The edges correspond to the outputs of the channel. Conditioned on the state $S=(i,j)$, when the transmitter chooses an input $x$ from the allowed subset, the output of the channel is an edge $y$ drawn randomly from the edges incident to $x$. The graph can be also viewed as the confusability graph of the underlying channel without the state-dependent input constraints.}
\label{fig:BKS}
\end{figure}

\section{Proof of Theorem \ref{thm:cyclic_graph}} \label{proof:C5}
Before we delve into the proof, let us first note that for the graph channel with state defined by the graph $\mathcal{C}_5$, it is equivalent to consider the channel with its input-output relationship captured by following algebraic form,
\begin{align} \label{eq:C5_algebraic}
  Y = X + S  \pmod 5
\end{align}
where the channel input $X\in \{0,1\}$, and channel state $S\in \{0,1,2,3,4\}$. It facilitates the proof to consider this algebraic form of the channel. In the following, it will be convenient to write $x \oplus_m y \triangleq x+y \pmod m$.

\subsection{Achievable rate by entanglement assistance} \label{proof:C5Q}
In this subsection, we show that there are entanglement-assisted coding schemes with causal CSIT, which achieve the rate $R= 1-H_b\big(\cos^2(\frac{\pi}{20})\big) \approx 0.8341$. 

Our scheme is based on a channel conversion strategy, which is inspired by quantum non-local games, adopted here to design the coding scheme of a channel with causal CSIT. Intuitively, we seek to convert the given $(A,S,X)\rightarrow (B,Y)$ channel with state $S$ and entanglement-assistance $AB$,  into a classical $X'\rightarrow Y'$ channel without state and without entanglement assistance.
Let us describe this channel conversion strategy in a more general way, since this will be also useful for other proofs in this paper. We assume that the transmitter and the receiver have access to a bunch of Bell pairs, each in the state $\ket{\Phi} = \frac{\sqrt{2}}{2}\big(\ket{00}+\ket{11}\big)$.

\noindent {\bf [Channel conversion]} 
Let $\xi_s \in [0,2\pi)$ and $\eta_y \in [0,2\pi)$ for $s\in \mathcal{S}$ and $y\in \mathcal{Y}$ be $|\mathcal{S}|+|\mathcal{Y}|$ angles. The values of these angles will be specified later.
Let ${\sf R}(\theta) \triangleq \ket{0}\!\bra{0} + e^{i\theta}\!\ket{1}\!\bra{1}$ denote the phase gate for a qubit with phase $\theta\in [0,2\pi)$.
The channel conversion strategy works as follows. 

For each time slot, the transmitter and the receiver consume one Bell pair. Prior to transmission, the transmitter applies the rotation gate ${\sf R}(-\xi_s)$ to its qubit conditioned on the channel state $S=s$, and then performs a local measurement in the $\{\ket{+}, \ket{-}\}$ basis, where $\ket{+}\triangleq \frac{\sqrt{2}}{2}\big(\ket{0}+\ket{1}\big)$ and $\ket{-}\triangleq \frac{\sqrt{2}}{2}\big(\ket{0}-\ket{1}\big)$.
The measurement outcome is mapped to a classical random variable $U\in \{0,1\}$ by assigning the outcome $\ket{+}$ to $0$ and $\ket{-}$ to $1$. 

The transmitter selects a binary symbol $X'\in \{0,1\}$ (which will serve as the input for the converted channel) and performs the classical pre-processing 
\begin{align}
	X=X'\oplus_2 U.\label{eq:Xprime}
\end{align}
The symbol $X$ is then transmitted over the original channel. 

Upon receiving the channel output $Y=y$, the receiver applies the rotation gate ${\sf R}(\eta_y)$ to its qubit, followed by a local measurement in the $\{\ket{+}, \ket{-}\}$ basis. Using the same mapping $\ket{+}$ to $0$ and $\ket{-}$ to $1$, the receiver obtains a classical binary output $Y'\in \{0,1\}$.

\begin{remark}
	A general channel conversion strategy need not be restricted to using the maximally entangled state or the rotation gates. Instead, the transmitter and the receiver may start from an arbitrary shared state and conduct measurements that depend on their respective inputs. The goal of a channel conversion strategy is to convert each use of the channel (with only the current channel state) to a `new' channel without state that potentially has a larger classical capacity. Channel conversion is also implicit in the classical capacity formula with causal CSIT,  $C^{\C}(\mN)  = \max_{\mP_{U}, x(u,s)} I(U;Y)$ where $x(u,s)$ is a memoryless map applied over each channel use, that `converts' the $(X,S)\rightarrow Y$ channel with state, into a $U\rightarrow Y$ channel without state.
\end{remark}

\noindent {\bf [Analysis]}
To analyze the behavior of the converted channel, consider $S=s, Y=y$.
The quantum state after both parties have done their operations is
$\ket{\Phi'}= \frac{\sqrt{2}}{2}\big(\ket{00}+e^{i(\eta_y-\xi_s)}\ket{11}\big)= \frac{\sqrt{2}}{4}\big(1+e^{i(\eta_y-\xi_s)}\big) \big( \ket{++}+\ket{--} \big)+ \frac{\sqrt{2}}{4}\big(1-e^{i(\eta_y-\xi_s)}\big) \big( \ket{+-}+\ket{-+} \big)$.
The probability that the measurement outcomes ($U$ at the transmitter, $Y'$ at the receiver) are identical is given by,
\begin{align}
  &\Pr(Y'=U\mid S=s, Y=y)\notag \\
  &= 2\times |\frac{\sqrt{2}}{4}\big(1+e^{i(\eta_y-\xi_s)}\big)|^2\\
  &= \cos^2\left(\frac{\eta_y-\xi_s}{2}\right), \label{eq:PrYpeqUcSY}
\end{align}
corresponding to the state collapsing to either $\ket{++}$ or $\ket{--}$. It follows that the measurement outcomes do not match with probability 
\begin{align}
	\Pr(Y'\neq U \mid S=s, Y=y)=\sin^2\left(\frac{\eta_y-\xi_s}{2}\right)\label{eq:PrYpneqUcSY}
\end{align}
corresponding to the state collapsing to either $\ket{+-}$ or $\ket{-+}$.

Moreover, since the reduced state of either subsystem of a Bell pair is maximally mixed, and a maximally mixed state after any unitary operation is still maximally mixed, $U$ is uniformly distributed over $\{0,1\}$ regardless of $s\in \mathcal{S}$. It follows that $X$ is independent of $S$. 

Next we wish to determine the value of $\Pr(Y'=X')$ that defines the resulting BSC. In order to do so we will first determine the conditional probabilities $\Pr(Y'=X'\mid X=0)$ and $\Pr(Y'=X'\mid X=1)$. First consider the $X=0$ case, and note that according to \eqref{eq:Xprime}, we have $(X=0)\iff(X'=U)$.
We therefore obtain,
\begin{align}
  &\Pr(Y'=X'\mid X=0) \notag\\
  &=\Pr(Y'=U \mid X=0)  \\
  &=\sum_s\Pr(S=s\mid X=0)\Pr(Y'=U \mid X=0, S=s)\\
  &=\sum_{s} \mP_S(s) \Pr(Y'=U \mid X=0,S=s) \label{eq:PrYpeqUcX0}\\
  &=\sum_{s} \mP_S(s) \Pr(Y'=U \mid S=s, Y=s)\label{eq:XStoSY1}\\
  &=\sum_{s} \mP_S(s) \cos^2\left(\frac{\eta_s-\xi_s}{2}\right) \label{eq:PgivenX0}
\end{align}
Step \eqref{eq:PrYpeqUcX0} used the fact that $X$ is independent of $S$. Step \eqref{eq:XStoSY1} used the fact that according to the channel input-output relationship specified in \eqref{eq:C5_algebraic}, we have $(X=0,S=s)\iff (S=s,Y=s)$. The last step used \eqref{eq:PrYpeqUcSY}.
Similarly,
\begin{align}
  &\Pr(Y'=X'\mid X=1) \notag\\
  &=\Pr(Y'\neq U \mid X=1) \\
  &=\sum_{s} \mP_S(s) \Pr(Y'\neq U \mid X=1, S=s) \label{eq:PrYpneqUcX1}\\
  &=\sum_{s} \mP_S(s) \Pr(Y'\neq U \mid S=s, Y=s\oplus_5 1)\label{eq:XStoSY2}\\
  &=\sum_{s} \mP_S(s) \sin^2\left(\frac{\eta_{s\oplus_5 1}-\xi_s}{2}\right) \label{eq:PgivenX1}
\end{align}

Now, let us specify the angles $\{\xi_s\}, \{\eta_y\}$ as,
\begin{align}
  \begin{bmatrix}
    \xi_0 & \xi_1 & \xi_2 & \xi_3 & \xi_4
  \end{bmatrix}
  =
   \frac{\pi}{10}
   \begin{bmatrix}
    0 & 8 & 16 & 4 & 12
  \end{bmatrix},  \label{eq:angles_y_C5}  \\
  \begin{bmatrix}
    \eta_0 & \eta_1 & \eta_2 & \eta_3 & \eta_4
  \end{bmatrix}
  =
   \frac{\pi}{10}
   \begin{bmatrix}
    1 & 9 & 17 & 5 & 13
  \end{bmatrix}. \label{eq:angles_s_C5}
\end{align}
We illustrate the geometry of these angles in Fig. \ref{fig:pentagon}. In polar coordinates ($r,\theta$), the points $\{(1, \xi_i)\}_{i=0}^4$ form the five vertices of a pentagon centered at the origin. $\{(1, \eta_j)\}_{j=0}^4$ form another pentagon, obtained by rotating the first one counterclockwise by $\pi/10$. Note that $\eta_s-\xi_s = \tfrac{\pi}{10}$, and $\eta_{s\oplus_5 1}-\xi_s =\tfrac{9\pi}{10} \pmod{2\pi}, \forall s \in \{0,1,\cdots, 4\}$.

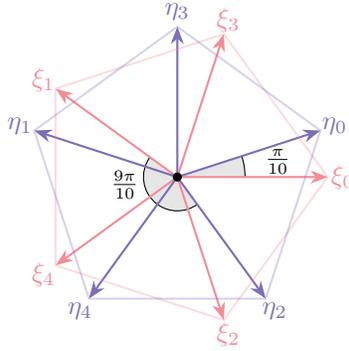
\begin{figure}[h]
\center
\begin{tikzpicture}[scale=2,>=Stealth]

\definecolor{myred}{RGB}{242,141,151}
\definecolor{mypurple}{RGB}{127,110,180}

\def\n{5}
\def\r{1}

\coordinate (A) at (1,0);
\coordinate (B) at (0,0);
\coordinate (C) at (18:\r);

\coordinate (D) at (72*2:\r);
\coordinate (E) at (72*4+18:\r);

\pic[draw,  fill=gray!20, font=\footnotesize, angle radius=9mm, "$\tfrac{\pi}{10}$", angle eccentricity=1.5] {angle = A--B--C};
\pic[draw, fill=gray!20, font=\footnotesize, angle radius=4.5mm] {angle = D--B--E};
\node at (-0.35,-0.05) {\footnotesize $\tfrac{9\pi}{10}$};

\foreach \i/\j  in {0/0,1/3,2/1,3/4,4/2} {
  \draw[->, thick, myred] (0,0) -- (72*\i:\r)
    node[font=\small, pos=1.1] {$\xi_{\j}$};
}

\foreach \i/\j  in {0/0,1/3,2/1,3/4,4/2} {
  \draw[->, thick, mypurple] (0,0) -- (72*\i+18:\r)
      node[font=\small, pos=1.1] {$\eta_{\j}$};
}

\draw[myred,opacity=.3,thick] (0:\r)  
        \foreach \i in {0,...,4} {  
            -- (72*\i:\r) 
        } -- cycle; 
        
\draw[mypurple,opacity=.3,thick] (18:\r)  
        \foreach \i in {0,...,4} {  
            -- (72*\i+18:\r)  
        } -- cycle; 

\fill[black] (0,0) circle (0.03);

\end{tikzpicture}
\caption{Plot of the angles $\{\xi_s\}, \{\eta_y\}$ for $\mathcal{C}_5$.}
\label{fig:pentagon}
\end{figure}

Plugging in the values of the angles $\{\xi_s\}, \{\eta_y\}$ into \eqref{eq:PgivenX0}, \eqref{eq:PgivenX1}, we obtain $\Pr(Y'=X'\mid X=0) = \cos^2(\pi/20)$ and $\Pr(Y'=X'\mid X=1) = \sin^2(9\pi/20)= \cos^2(\pi/20)$. It follows that $\Pr(Y'=X') = \cos^2(\pi/20)$. The above analysis holds for $X'\in\{0,1\}$. Therefore, $X'\to Y'$ is a BSC with the correct transmission probability equal to $\cos^2(\pi/20)$. A classical code over this BSC then achieves the rate $1-H_b(\cos^2(\pi/20)) \approx 0.8341$. \hfill \qed

\subsection{Classical capacity}
In this subsection, we prove that $C^{\C}(\mathcal{C}_5) = 0.8$. 
We apply the formula for the classical capacity with causal CSIT (e.g., \cite[Thm. 7.2]{NIT}).
\begin{align}
  C^{\C}(\mathcal{C}_5) & = \max_{\mP_{U}, x} \big( \underbrace{H(Y) - H(Y\mid U)}_{I(U;Y)} \big) \label{eq:derive_C5_C_1} \\
  &\leq \log_2(5)-\min_{x} \min_u H(Y\mid U=u) \label{eq:derive_C5_C_2}\\
  &\leq \log_2(5) - H\big(\big[\tfrac{2}{5}, \tfrac{2}{5}, \tfrac{1}{5}\big]\big) \label{eq:derive_C5_C_3}\\
  &= 0.8
\end{align}
In Step \eqref{eq:derive_C5_C_1}, the maximization is over all distributions $\mP_U$ of an auxiliary random variable $U$ over the support $\mathcal{U}$, and over all functions $x \colon \mathcal{U}\times \{0,1,2,3,4\} \to \{0,1\}$.
Step \eqref{eq:derive_C5_C_2} is because $H(Y) \leq \log_2 |\mathcal{Y}| = \log_2 5$.
To see Step \eqref{eq:derive_C5_C_3}, fix any function $x$ and $U=u$, and then $Y=x(u,S) \oplus_5 S$. Let $y_{s} \triangleq x(u,s)\oplus_5 s$. Observe that in the vector $[y_0,y_1,y_{2},y_{3},y_{4}]$ an element cannot repeat more than twice since  $x(u,s) \in \{0,1\}$. This restricts the possible $\mP_Y$ (up to permutations) to $[\tfrac{2}{5},\tfrac{2}{5},\tfrac{1}{5},0, 0]$, $[\tfrac{2}{5},\tfrac{1}{5},\tfrac{1}{5},\tfrac{1}{5}, 0]$, or $[\tfrac{1}{5},\tfrac{1}{5},\tfrac{1}{5},\tfrac{1}{5}, \tfrac{1}{5}]$, and thus we have $H(Y\mid U=u) \geq H\big(\big[\tfrac{2}{5}, \tfrac{2}{5}, \tfrac{1}{5}\big]\big)$ for all $x, u$. This proves that $C^{\C}(\mathcal{C}_5)\leq 0.8$.

To prove that $C^{\C}(\mathcal{C}_5)\geq 0.8$, observe that $I(U;Y)=0.8$ is achieved by choosing $U$ to be uniformly distributed over $\{0,1\}$, and $x(u,s) = u\oplus_2 s$ for $u\in \{0,1\}$ and $s\in \{0,1,2,3,4\}$.  \hfill \qed

\section{Proof of Theorem \ref{thm:activation}} \label{proof:activation}
Recall that we are given a B-KS set in $\mathbb{C}^d$ specified by $(\{\mathcal{A}_i\}_{i=1}^a,\{\mathcal{B}_{j}\}_{j=1}^{b})$. To define the channel with state, we view the channel with state as imposing input constraints on a channel $\mN_o$ (without state) as explained in Remark \ref{rem:input_constraint}. 
The input alphabet of $\mN_o$ is defined as $\mathcal{X} = \mathcal{X}^A \cup \mathcal{X}^B$ where $\mathcal{X}^A \triangleq \{(1,i,k) \colon i\in [a], k\in [d]\}$ and $\mathcal{X}^B \triangleq  \{(2,j,k) \colon j\in [b], k\in [d]\}$ such that $|\mathcal{X}| = |\mathcal{X}^A|+|\mathcal{X}^B| = ad+bd= (a+b)d$. For each $i\in [a]$, let $\mathcal{X}^A_i \triangleq \{(1,i,k)\colon k\in [d]\}$, and for each $j\in [b]$, let $\mathcal{X}^B_j \triangleq \{(2,j,k)\colon k\in [d]\}$.

For each $i\in [a]$, we use $\ket{\phi_{(1,i,k)}}$ to denote the (normalized) $k^{th}$ vector ($k\in [d]$) in the basis $\mathcal{A}_i$, and for each $j\in [b]$, we use $\ket{\phi_{(2,j,k)}}$ to denote the (normalized)  $k^{th}$ ($k\in [d]$) vector in the basis $\mathcal{B}_j$.
The output alphabet of $\mN_o$ is defined as $$\mathcal{Y} \triangleq \big\{ \{x,x'\} \colon x\in \mathcal{X}, x'\in \mathcal{X}, \mbox{~s.t.~} \ket{\phi_x} \mbox{~and~} \ket{\phi_{x'}} \mbox{~are~orthogonal} \big\}.$$ The channel $\mN_o(y\mid x) > 0$ if and only if $x\in y$. The set of allowed inputs depends on a random state $S \in \mathcal{S} \triangleq \{(i,j) \colon i\in [a], j\in [b]\}$, such that when the channel state $S=(i,j)$, the allowed channel inputs are exactly the $2d$ elements in $\mathcal{X}^A_i \cup \mathcal{X}^B_j$.  Although we have not specified the exact values of $\mN_o(y\mid x)$ or $\mP_S(s)$, any implementation that satisfies the constraints will suffice for our purpose.

 \subsection{Proof of $c_0^{\C} = C_0^{\C} = 0$}
As a first step, we show that $M_{1,\opt}^{\C} < 2$ so $M_{1,\opt}^{\C} = 1$ and thus $c_0^{\C} = 0$. We prove it by contradiction. Suppose $M_{1,\opt}^{\C} \geq 2$, then there exists a deterministic (Remark \ref{rem:det_coding}) classical $(M=2,n=1)$ coding scheme with zero error.   For one use of the channel, conditioned on $S=(i,j)$, let us denote the output of this encoder as $u_{i,j} \in \mathcal{X}^A_i \cup \mathcal{X}^B_{j}$ if the message $W=1$, and $v_{i,j} \in \mathcal{X}^A_i \cup \mathcal{X}^B_{j}$ if the message $W=2$. Let $\mathcal{U} = \{u_{i,j}\colon (i,j)\in \mathcal{S}\}$ and $\mathcal{V} = \{v_{i,j}\colon (i,j)\in \mathcal{S}\}$.
Clearly, for $u\in \mathcal{U}$ and $v\in \mathcal{V}$, we must have $u \not= v$, as different messages have to correspond to different inputs in order for a zero-error decoding. 
Besides, $\ket{\phi_u}$ and $\ket{\phi_v}$ cannot be orthogonal, because otherwise $\mN_o(\{u,v\}\mid u) \mN_o(\{u,v\}\mid v) > 0$, which means both $W=1$ and $W=2$ could lead to $y=\{u,v\}$,  thus making zero-error decoding impossible.
This implies that, for each $i\in [a]$ and $j\in [b]$, $u_{i,j}$ and $v_{i,j}$ cannot  both belong to $\mathcal{X}^A_i$ or both belong to $\mathcal{X}^B_{j}$.

Without loss of generality, suppose $u_{1,1} \in \mathcal{X}^A_1$ and $v_{1,1} \in \mathcal{X}^B_1$. 
It follows that for each $j =\{2,3,\cdots, b\}$, $v_{1,j} \not\in \mathcal{X}^A_1$ so $ v_{1,j} \in \mathcal{X}^B_{j}$, and thus $u_{1,j} \in \mathcal{X}^A_1$. Due to similar reasons, for each $i \in \{2,3,\cdots, a\}$ and $j \in [b]$, we have $u_{i,j} \in \mathcal{X}^A_i$ and $v_{i,j}\in \mathcal{X}^B_{j}$.
Therefore, $\ket{\phi_{u_{1,1}}}\in \mathcal{A}_1, \ket{\phi_{u_{2,1}}} \in \mathcal{A}_2, \cdots, \ket{\phi_{u_{a,1}}} \in \mathcal{A}_a$, $\ket{\phi_{v_{1,1}}}\in \mathcal{B}_1, \ket{\phi_{v_{1,2}}} \in \mathcal{B}_2, \cdots, \ket{\phi_{v_{1,b}}} \in \mathcal{B}_b$. Due to the definition of a B-KS set, there exist $i \in [a], j\in[b]$ such that $\ket{\phi_{u_{i,1}}}$ and $\ket{\phi_{v_{1,j}}}$ are orthogonal, but since $u_{i,1}\in \mathcal{U}$ and $v_{1,j} \in \mathcal{V}$, those two vectors cannot be orthogonal. This contradiction proves that there is no classical $(M=2,n=1)$ coding scheme with zero error, so $M^{\C}_{1,\opt} = 1$.

The proof is now completed by invoking the following lemma, which works for a general channel with state.
\begin{lemma} \label{lem:singleletter_classical}
  Given any channel with state, if $M^{\C}_{1,\opt} = 1$, then $M^{\C}_{n,\opt} = 1$ for all $n\in \mathbb{N}$. Equivalently, $c_0^{\C}=0 \Leftrightarrow C_0^{\C} = 0$. 
\end{lemma}
A formal proof of Lemma \ref{lem:singleletter_classical} is provided in Appendix \ref{proof:singleletter_classical}. Intuitively, the reasoning behind Lemma \ref{lem:singleletter_classical} can be understood as follows. Consider a one-bit message, i.e., $M=2$. Assume $M^{\C}_{1,\opt} = 1$, which means one channel-use is not enough to perfectly resolve the receiver's confusion between $W=1$ and $W=2$. Now suppose $k$ channel-uses are also not enough for this task, i.e., $M^{\C}_{k,\opt} = 1$, then since after $k$ channel uses, there is still a possible scenario such that the receiver remains confused between $W=1$ and $W=2$, the remaining task of eliminating that confusion over the next, i.e., $(k+1)^{th}$ channel-use is as hard as the original task of eliminating that confusion over the first channel-use, which is already assumed infeasible. Therefore $M^{\C}_{k+1,\opt} = 1$. Inductive reasoning then implies $M^{\C}_{n,\opt} = 1$ for all $n\in \mathbb{N}$. 
\hfil \qed

\subsection{Proof of $c_0^{\EA} = C_0^{\EA} = 1$}
Firstly, we show how to construct an $(M=2, n=1)$ coding scheme with zero error, which then implies $c_0^{\EA} \geq 1$ and thus $C_0^{\EA} \geq 1$.  Note that the scheme uses the channel only once, and sends a message $W\in \{1,2\}$. The idea builds upon \cite[Thm. 2]{cubitt2010improving}, developed further here for the channel with state setting.

The transmitter and the receiver use a maximally entangled state of dimension $d$, i.e., $\ket{\Phi} = \frac{1}{\sqrt{d}}\sum_{i=0}^{d-1}\ket{i}\ket{i}$. 
Note that conditioned on each channel state $S=(i,j)$, the transmitter should pick an input from either $\mathcal{X}^A_i$ or $\mathcal{X}^B_{j}$. The scheme is such that, if $W=1$, the transmitter measures its quantum system on the basis $\mathcal{A}_i$, and the state collapses to $\ket{\phi_x}\in \mathcal{A}_i$ for some $x \in \mathcal{X}^A_i$;  if $W=2$, the transmitter measures its quantum system on the basis $\mathcal{B}_{j}$, and the state collapses to $\ket{\phi_x}\in \mathcal{B}_j$ for some $x \in \mathcal{X}^B_j$. The input $X=x$ is then sent to the channel.
According to the ``ricochet" property \cite[Ex. 3.7.12]{Wilde_2017} of the maximally entangled state $\ket{\Phi}$, we have $\ket{\Phi} = (U U^\dagger \otimes I)\ket{\Phi} = (U\otimes U^*)\ket{\Phi}$ for any unitary matrix $U\in \mathbb{C}^{d\times d}$. Let $U$ be the unitary matrix whose columns are the basis vectors with respect to which the transmitter performs the measurement. It follows that the receiver's state collapses to $\ket{\phi_x^*}$, i.e., the complex conjugate of $\ket{\phi_x}$.

At the receiver, suppose the channel outputs $y = \{x, x'\}$. Then according to the definition of the channel, $\ket{\phi_x}$ and $\ket{\phi_{x'}}$ are orthogonal, and so are $\ket{\phi_x^*}$ and $\ket{\phi_{x'}^*}$. The receiver conducts a projection-valued measure (PVM) $\{\ket{\phi_x^*}\!\!\bra{\phi_x^*}, \ket{\phi_{x'}^*}\!\!\bra{\phi_{x'}^*}, I-\ket{\phi_x^*}\!\!\bra{\phi_x^*}- \ket{\phi_{x'}^*}\!\!\bra{\phi_{x'}^*}\}$ on its quantum system, which allows it to distinguish between $x$ and $x'$ with certainty. Therefore, the receiver can always identify the correct input $x$. Note that the input $x$ used for $W=1$ is some $x\in  \mathcal{X}^A$, whereas the input $x$ used for $W=2$ is some $x\in \mathcal{X}^B$, and since $\mathcal{X}^A\cap \mathcal{X}^B=\emptyset$, the message can be decoded from $x$ with certainty.

Next, we prove that $C^{\EA}_0  \leq 1$ and thus $c^{\EA}_0\leq 1$. For this purpose, we invoke a result of \cite{beigi2010entanglement}, which states that for any classical channel $\mN(y\mid x)$ (without state), the zero-error capacity is upper bounded by the logarithm of the Lov{\'a}sz $\vartheta$ function of the channel's confusability graph. To apply the result, note that the zero-error capacity of our channel with state is upper bounded by that of the channel for the worst case of the state, and thus any case, e.g., $S=(1,1)$. When the state is $(1,1)$, the input alphabet of the channel is $\mathcal{X}^A_1\cup \mathcal{X}^B_1$. To find the value of the Lov{\'a}sz $\vartheta$ function of the confusability graph corresponding to the channel state being fixed to $S=(1,1)$, let us construct the graph, by having its $2d$ vertices in $\mathcal{X}^A_1\cup \mathcal{X}^B_1$, such that any two distinct vertices $x, x' \in \mathcal{X}^A_1\cup \mathcal{X}^B_1$ are connected by an edge if and only if $\ket{\phi_x}$ and $\ket{\phi_{x'}}$ are orthogonal. Since $\ket{\phi_x}$ and $\ket{\phi_{x'}}$ are orthogonal for $x\not=x'$ as long as $x$ and $x'$ both come from $\mathcal{X}^A_1$ or both come from $\mathcal{X}^B_1$, the vertices in $\mathcal{X}^A_1$ form a clique (complete subgraph) of size $d$ and vertices in $\mathcal{X}^B_1$ form another clique of size $d$. Therefore, the confusability graph can be partitioned into two cliques. Since the value of the Lov{\'a}sz $\vartheta$ function of such a graph is at most $2$ (e.g., \cite{riddle2003sandwich, lovasz2019graphs}), it follows that the entanglement-assisted zero-error capacity is upper bounded by $\log_2 (2) = 1$. \hfil \qed

\section{Conclusion}
In contrast to the classical point-to-point channel (without state) where prior work has shown that entanglement-assistance cannot improve the capacity or activate the zero-error capacity, we show that for point-to-point channels with causal CSIT, entanglement-assistance can in some cases improve the capacity, and in some cases activate the zero-error capacity. Graph channels with state emerge as an interesting class of channels, not only because they demonstrate the aforementioned quantum advantage in capacity, but also because these channels are deterministic (output is determined by input and state) which might make them more tractable for further studies, for example to determine their precise entanglement-assisted capacity. The largest  gain in capacity we found in this work is $\tfrac{12}{\pi^2}\approx 21.6\%$. It remains an important open question to determine whether much larger gains are possible, and ultimately to determine the largest multiplicative capacity gain possible from entanglement assistance across all point-to-point channels with causal CSIT. Our work focused on causal CSIT. In contrast, for non-causal CSIT, we do not know whether entanglement assistance can improve capacity. This is another promising direction for future work.

\appendix
\section{Proof of Theorem \ref{thm:complete_graph}} \label{proof:complete_graph}
In this section we study the graph channel with state specified by $(\mathcal{K}_m, \unif)$. Without loss of generality, let the channel output alphabet be $\mathcal{Y} = \{0,1,2,\cdots, m-1\}$, and the channel state alphabet be $\mathcal{S} = \{(s_0,s_1)\colon s_0, s_1 \in \mathcal{Y}, s_0 < s_1\}$. 
As an example, for $m=4$, we have $\mathcal{Y}=\{0,1,2,3\}$ and $\mathcal{S} = \{(0,1),(0,2),(0,3),(1,2),(1,3),(2,3)\}$.
\subsection{Achievable rate by entanglement assistance}
In this subsection we prove that there are entanglement-assisted coding schemes with causal CSIT, which achieve the rate $R^{\EA}(\mathcal{K}_m,p)$ in \eqref{eq:RQ_Km}.
We adopt the same channel conversion strategy used in Section \ref{proof:C5Q}. We refer the readers to the  description of the channel conversion strategy and its analysis in \eqref{eq:PrYpeqUcSY}--\eqref{eq:PrYpneqUcX1}.

To apply the channel conversion strategy, let us specify the angles $\xi_s$ for $s\in \mathcal{S}$ and $\eta_y$ for $y\in \mathcal{Y}$.
In the following, for $x \in \mathbb{R}$, let $(x)_{2\pi} \triangleq x \pmod{2\pi} \in [0,2\pi)$. For $a\neq b\in [0,2\pi)$, define the function
\begin{align} \label{eq:def_arrow_function}
  \uparrow_b^a ~\triangleq \begin{cases}
    \big(\tfrac{a+b+\pi}{2}\big)_{2\pi} & \mbox{if}~ a> b \\
    \big(\tfrac{a+b-\pi}{2}\big)_{2\pi} & \mbox{if}~ a< b
  \end{cases},
\end{align}
which maps $[0,2\pi)\times [0,2\pi) \to [0,2\pi)$. To see the geometric meaning of $\uparrow_b^a$, consider $a,b$ as two (different) unit arrows with directions specified by the angles $a$ and $b$, respectively. Then $\uparrow_b^a$ defines the angle perpendicular to the angular midpoint of  $a$ and $b$, with the orientation chosen to be closer to $a$. It follows that $(\uparrow_{b}^a + \uparrow_{a}^b)=a+b$, and $(\uparrow_{a}^b - \uparrow_{b}^a) \in \{\pi,-\pi\}$.

For $y\in \{0,1,\cdots, m-1\}$, let
\begin{align}
  \eta_y = \tfrac{2\pi}{m} y,
\end{align}
and for $s=(s_0,s_1) \in \mathcal{S}$, let
\begin{align}
  \xi_s = ~\uparrow_{\eta_{s_1}}^{\eta_{s_0}}.
\end{align}
Geometrically, $\{\eta_y\}$ are the $m$ angles that uniformly partition the interval $[0,2\pi)$. For each $s=(s_0,s_1)\in \mathcal{S}$, the angle $\xi_s$ is defined as the direction perpendicular to the angular midpoint of  $\eta_{s_0}$ and $\eta_{s_1}$, with the orientation chosen to be closer to $\eta_{s_0}$. See Fig. \ref{fig:K4_angles} for an illustration of the case with $m=4$.

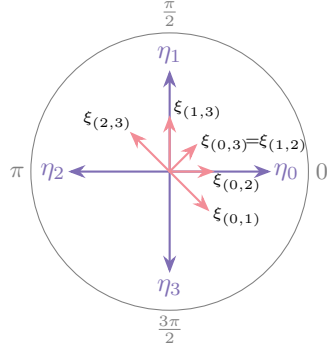
\begin{figure}[h]
\center
\begin{tikzpicture}[scale=1.5, line cap=round, line join=round, >=Stealth]
\definecolor{myred}{RGB}{242,141,151}
\definecolor{mypurple}{RGB}{127,110,180}
\begin{scope}
  \draw[gray] (0,0) circle (1.23);
  \node[gray, font=\footnotesize] at (1.35,0) {$0$};
  \node[gray, font=\footnotesize] at (0,1.4) {$\tfrac{\pi}{2}$};
  \node[gray, font=\footnotesize] at (-1.35,0) {$\pi$};
  \node[gray, font=\footnotesize] at (0,-1.4) {$\tfrac{3\pi}{2}$};

  \foreach \ang\j\l in {0/0/0.9,90/1/0.9,180/2/0.9,270/3/0.9}{
    \draw[->, thick, mypurple] (0,0) -- (\ang:\l)
    node[font=\small, pos=1.15] {$\eta_{\j}$};
  }

  \foreach \ang\l in {315/0.5,0/0.4,45/0.35,90/0.5,135/0.5}{
    \draw[->, thick, myred] (0,0) -- (\ang:\l);
  }
  
  \node[font=\tiny, rotate=0] at (0.6,-0.4) {$\xi_{(0,1)}$};
  \node[font=\tiny, rotate=0] at (0.6,-0.1) {$\xi_{(0,2)}$};
  \node[font=\tiny, rotate=0] at (0.75,0.25) {$\xi_{(0,3)}\!\!\!= \!\!  \xi_{(1,2)}$};
  \node[font=\tiny, rotate=0] at (0.25,0.55) {$\xi_{(1,3)}$};
  \node[font=\tiny, rotate=0] at (-0.55,0.45) {$\xi_{(2,3)}$};
\end{scope}
 
\end{tikzpicture}
\caption{Plot of the angles $\{\xi_s\}, \{\eta_y\}$ for $K_4$.}
\label{fig:K4_angles}
\end{figure}

Let us first consider the special case with $p=1$. We have that for $s=(s_0,s_1) \in \mathcal{S}$, $(X=0,S=s) \Longleftrightarrow (S=s,Y=s_0)$.
It then follows that,
\begin{align}
  &\Pr(Y'=X'\mid X=0)\notag \\
  &=\sum_{s} \mP_S(s)   \Pr(Y'=U\mid S=s, Y=s_0) \label{eq:use_prob1}  \\
  &=\frac{1}{\binom{m}{2}} \sum_{s=(s_0,s_1)\in \mathcal{S}} \cos^2\left( \frac{\eta_{s_0}-\xi_s}{2} \right) \label{eq:use_cos} \\
  &=\tfrac{2}{m(m-1)} \sum_{s=(s_0,s_1)\in \mathcal{S}} \cos^2\left( \frac{\eta_{s_0}-~\uparrow_{\eta_{s_1}}^{\eta_{s_0}}}{2} \right) \\
  &=\tfrac{2}{m(m-1)} \sum_{s=(s_0,s_1)\in \mathcal{S}} \cos^2\left( \frac{\eta_{s_0}- (\tfrac{\eta_{s_0}+\eta_{s_1}- \pi}{2})_{2\pi}}{2} \right) \label{eq:interm1}  \\
  &=\tfrac{2}{m(m-1)} \sum_{s=(s_0,s_1)\in \mathcal{S}} \cos^2\left( \frac{\eta_{s_0}- (\tfrac{\eta_{s_0}+\eta_{s_1}- \pi}{2})}{2} \right) \label{eq:period}  \\
  &= \tfrac{2}{m(m-1)} \sum_{s=(s_0,s_1)\in \mathcal{S}} \cos^2\left( \frac{\left(m - 2(s_1-s_0)\right)\pi}{4m} \right)  \\
  &=\tfrac{2}{m(m-1)} \sum_{\ell=1}^{m-1} (m-\ell)  \cos^2\left( \frac{(m-2\ell)\pi}{4m}\right)\label{eq:subd}\\
  &=\tfrac{1}{m(m-1)} \sum_{k=1}^{m-1} k  \left(1+\cos \left( \frac{(2k-m)\pi}{2m}\right)\right)  \label{eq:subk} \\
  &=\frac{1}{2} + \tfrac{1}{m(m-1)}  \sum_{k=1}^{m-1}  k \sin\left(\frac{k\pi}{m}\right)   \\
  &=\underbrace{\frac{1}{2} + \frac{\cot(\tfrac{\pi}{2m})}{2(m-1)}}_{\triangleq r_m}    \label{eq:interm2}
\end{align}
Step \eqref{eq:use_prob1} uses \eqref{eq:PrYpeqUcX0}, and Step \eqref{eq:use_cos} uses \eqref{eq:PrYpeqUcSY}.
Step \eqref{eq:interm1} is due to $s_0<s_1$, so $\eta_{s_0}<\eta_{s_1}$, and the definition \eqref{eq:def_arrow_function}. Step \eqref{eq:period} is because the function $\cos^2(x)$ has period $\pi$. Step \eqref{eq:subd} follows by substituting $\ell=s_1-s_0$ which, for $(s_0,s_1)\in\mathcal{S}$, can only takes values in $\{1,2,\cdots,m-1\}$, and noting that there are $m-\ell$ pairs $(s_0,s_1)\in\mathcal{S}$ such that $\ell=s_1-s_0$. Step \eqref{eq:subk} makes the substitution $k=m-\ell$ and uses the property $2\cos^2(x)=1+\cos(2x)$.
To see Step \eqref{eq:interm2}, note that $\sum_{k=1}^{m-1}k \sin(\frac{k\pi}{m})$ is the imaginary part of the sum $\sum_{k=1}^{m-1}  k e^{i\frac{k\pi}{m}}=\frac{1}{(1-e^{i\frac{k\pi}{m}})^2} \big( (2-m) e^{i\frac{k\pi}{m}}+m\big) = \frac{(2-m)  +me^{-i\frac{k\pi}{m}}}{-4\sin^2(\frac{k\pi}{2m})}$. The imaginary part is equal to $\frac{m\sin(\frac{k\pi}{m})}{4\sin^2(\frac{k\pi}{2m})}=\frac{m}{2}\cot(\frac{k\pi}{4m})$.

Similarly, for $s=(s_0,s_1) \in \mathcal{S}$, $(X=1,S=s) \Longleftrightarrow (S=s,Y=s_1)$, and it follows that,
\begin{align}
  &\Pr(Y'=X'\mid X=1)\notag \\
  &=\sum_{s} \mP_S(s) \Pr(Y' \not= U\mid S=s, Y=s_1) \label{eq:use_prob2} \\
  &=\frac{1}{\binom{m}{2}} \sum_{s=(s_0,s_1)\in \mathcal{S}} \sin^2\big( \frac{\eta_{s_1}-\xi_s}{2} \big) \label{eq:use_sin} \\
  &=\tfrac{2}{m(m-1)} \sum_{s=(s_0,s_1)\in \mathcal{S}} \sin^2\big( \frac{\eta_{s_1}-~\uparrow_{\eta_{s_1}}^{\eta_{s_0}}}{2} \big) \\
  &=\tfrac{2}{m(m-1)} \sum_{s=(s_0,s_1)\in \mathcal{S}} \cos^2\big( \frac{\eta_{s_0}-~\uparrow_{\eta_{s_1}}^{\eta_{s_0}}}{2} \big) \label{eq:sin_to_cos} \\
  &=r_m ~~ \because \eqref{eq:interm1}\mbox{--}\eqref{eq:interm2}
\end{align}
Step \eqref{eq:use_prob2} uses \eqref{eq:PrYpneqUcX1}, and Step \eqref{eq:use_sin} uses \eqref{eq:PrYpneqUcSY}.
To see Step \eqref{eq:sin_to_cos}, let $a = \eta_{s_0}$ and $b= \eta_{s_1}$. Since $b-{\uparrow_{b}^a} = -a+{\uparrow_a^b} = -a+{\uparrow_b^a} \pm \pi$, we have $\sin^2(\tfrac{b-\uparrow_b^a}{2}) = \sin^2(\tfrac{-a+\uparrow_b^a \pm \pi}{2}) = \cos^2(\tfrac{-a+\uparrow_b^a}{2}) = \cos^2(\tfrac{a-\uparrow_b^a}{2})$.

It follows that $\Pr(Y'=X') = r_m$. The above analysis holds for $X'\in\{0,1\}$. Therefore, $X'\to Y'$ is a BSC with the correct transmission probability equal to $r_m$. 

Given $p<1$, suppose we use the same channel conversion strategy. Then the converted channel has probability $p$ of being the same as the above. With probability $(1-p)$, the converted channel becomes the completely noisy BSC with $\Pr(Y'=X') = 0.5$. Taking the average, $X'\to Y'$ is a BSC with correct transmission probability equal to $pr_m+\tfrac{1-p}{2}=q_{m,p}$. A classical code over this BSC then achieves the rate $1-H_b(q_{m,p})$. \hfill \qed

\subsection{Classical capacity}
In this subsection, we derive $C^{\C}(\mathcal{K}_m,p)$ as in \eqref{eq:CC_Km}. Applying the formula for the classical capacity with causal CSIT, we obtain
\begin{align}
  C^{\C}(\mathcal{K}_m,p) &= \max_{\mP_{U}, x} \big( H(Y) - H(Y\mid U) \big) \label{eq:derive_Km_C_1} \\
  &\leq \log_2 m -  \min_{u} \min_{x}  H(Y\mid U=u)  \label{eq:to_minimize}
\end{align}
In Step \eqref{eq:derive_Km_C_1}, the maximization is over all distribution $\mP_U$ of an auxiliary random variable $U$ over the support $\mathcal{U}$, and over all functions $x \colon \mathcal{U}\times \mathcal{S} \to \{0,1\}$.
Note that in Step \eqref{eq:to_minimize}, for a fixed $u$, the minimization $\min_{x}  H(Y\mid U=u)$ is over all mappings $x(u,s) = \phi(s)$ where $\phi \colon \mathcal{S} \to \{0,1\}$. The conditional entropy $H(Y\mid U=u)$ is found with respective to $\phi$ as,
\begin{align} 
&H(Y\mid U=u) \notag\\
&=H\big( p\big[v^{\phi}_0,v^{\phi}_1,\cdots, v^{\phi}_{m-1}\big] + (1-p)\big[\tfrac{1}{m},\tfrac{1}{m},\cdots, \tfrac{1}{m}\big] \big)\\
&\triangleq h_{\phi} \label{eq:def_h_phi}
\end{align}
where
\begin{equation*}
    v^{\phi}_y = \tfrac{2}{m(m-1)} \sum_{s=(s_0,s_1)\in \mathcal{S}} \mathbb{I}[s_{\phi(s)}=y],~\forall y\in \mathcal{Y}.
   \end{equation*}
To see this, recall that the channel is a noisy version of the channel defined by $(\mathcal{K}_m,\unif)$ and $p\in [0,1]$. Also note that $v_y^{\phi}$ is the probability of $Y=y$ for the noiseless version (i.e., $p=1$). Therefore, the probability of $Y = y$ in the noisy version is equal to $pv_y^{\phi}+(1-p)\tfrac{1}{m}$.

We then prove the follow lemma.
\begin{lemma} \label{lem:minimization}
With the definition of $h_{\phi}$ in \eqref{eq:def_h_phi},
   \begin{align}
    &\min_{\phi} h_{\phi}  = \log_2 m - \frac{1}{m}\sum_{k=1}^m \big( 1+  pc_{m,k}  \big)\log_2 \big( 1+  pc_{m,k}  \big)\label{eq:minimum}
   \end{align}
   where $c_{m,k} \triangleq  \tfrac{m-2k+1}{m-1}$ and the minimization is over all functions $\phi \colon \mathcal{S} \to \{0,1\}$. 
\end{lemma}
\begin{proof}
We first observe that for each $\phi$ in the domain, $[v^{\phi}_0,v^{\phi}_1,\cdots, v^{\phi}_{m-1}]$ constitute a probability vector, and for $y\in \mathcal{Y}$, $v^{\phi}_y$ is proportional to the number of $s$ for which $s_{\phi(s)}=y$. 

Given any $\phi\colon \mathcal{S} \to \{0,1\}$, in the following we will construct a sequence of functions $$\phi_1,\phi_2,\cdots, \phi_{m-1}$$ in the domain, with $h_{\phi}\geq h_{\phi_1}\geq \cdots \geq h_{\phi_{m-1}} = \big(\mbox{RHS~of~}\eqref{eq:minimum}\big)$, thus proving the lemma. Let $y_1 = \argmax_{y} v_{y}^{\phi}$ denote the index corresponding to the largest value in $[v_0^{\phi}, v_1^{\phi}, \cdots, v_{m-1}^{\phi}]$. If the maximum is achieved by multiple values of $y$, we select the smallest such $y$.

Let us define $\phi_1$ such that for each $s\in \mathcal{S}$,  $\phi_1(s) = 1-\phi(s)$ when $y_1 \in s$ and $s_{\phi(s)} \neq y_1$; otherwise, set $\phi_1(s)=\phi(s)$. Note that for all $s\in \mathcal{S}$ that contains $y_1$, we have $s_{\phi_1(s)}=y_1$, and thus $v^{\phi_1}_{y_1} = \frac{2}{m(m-1)}\cdot (m-1)$.

Let us point out that $[v_0^{\phi_1}, v_1^{\phi_1}, \cdots, v_{m-1}^{\phi_1}]$, compared to $[v_0^{\phi}, v_1^{\phi}, \cdots, v_{m-1}^{\phi}]$, is obtained by sequentially transferring a non-negative amount of probability mass from some place indexed by $y\neq y_1$ to the place indexed by $y_1$. 
Since by definition $v_{y_1}^{\phi}$ is the largest value in the vector $[v_0^{\phi}, v_1^{\phi}, \cdots, v_{m-1}^{\phi}]$, it follows that $[v_0^{\phi_1}, v_1^{\phi_1}, \cdots, v_{m-1}^{\phi_1}]$  majorizes\footnote{Let $x=[x_1,\cdots, x_m]$ and $y=[y_1,\cdots, y_m]$ be two probability vectors, i.e., $x_i\geq 0, y_i \geq 0$ for all $i\in [m]$ and $\sum_{i=1}^m x_i = \sum_{i=1}^m y_i =1$. Let $x^{\downarrow}$ and $y^{\downarrow}$ denote the vectors obtained by rearranging the components of $x$ and $y$ in non-increasing order. We say that $x$ majorizes $y$ if and only if $\sum_{i=1}^k x_i^{\downarrow} \geq \sum_{i=1}^k y_i^{\downarrow}$ for $k\in \{1,2,\cdots, m-1\}$.} $[v_0^{\phi}, v_1^{\phi}, \cdots, v_{m-1}^{\phi}]$. Therefore, \\$[\bar{v}_0^{\phi_1}, \bar{v}_1^{\phi_1}, \cdots, \bar{v}_{m-1}^{\phi_1}]$ majorizes $[\bar{v}_0^{\phi}, \bar{v}_1^{\phi}, \cdots, \bar{v}_{m-1}^{\phi}]$, where $\bar{v} \triangleq pv+(1-p)\tfrac{1}{m}$. It follows that $h_{\phi}\geq h_{\phi_1}$ as Shannon entropy is Schur-concave \cite{Marshall_Olkin_Arnold_book}.

Then, let $y_2 = \argmax_{y \in \mathcal{Y} \setminus \{y_1\}} v_{y}^{\phi_1}$. For each $s\in \mathcal{S}$, define $\phi_2(s) = 1-\phi_1(s)$ when  $y_1 \not\in s, y_2 \in s$ and $s_{\phi_1(s)} \neq y_2$; otherwise, set $\phi_2(s)=\phi_1(s)$. Note that for all $s\in \mathcal{S}$ that does not contain $y_1$ but contains $y_2$, we have $s_{\phi_2(s)}=y_2$, and thus $v^{\phi_2}_{y_2} = \tfrac{2}{m(m-1)}\cdot (m-2)$. Due to similar reasons, we have $h_{\phi_1} \geq h_{\phi_2}$. 
 
Continue this process such that for $i\in \{3,\cdots, m-1\}$, in the $i^{th}$ step, let $$y_i = \argmax_{y \in \mathcal{Y} \setminus \{y_1,\cdots, y_{i-1}\}} v_{y}^{\phi_{i-1}},$$ and for each $s\in \mathcal{S}$, define $\phi_i(s) = 1-\phi_{i-1}(s)$ when $y_{k} \not\in s, \forall k \in \{1,\cdots, i-1\}, y_i \in s$ and $s_{\phi_{i-1}(s)} \neq y_i$; otherwise, set $\phi_i(s)=\phi_{i-1}(s)$. After the $(m-1)^{th}$ step, we obtain $\phi_{m-1}$, for which we have
\begin{align}
	&[v_{y_0}^{\phi_{m-1}}, v_{y_1}^{\phi_{m-1}},\cdots, v_{y_{m-1}}^{\phi_{m-1}},  v_{y_m}^{\phi_{m-1}}] \notag \\
	&= \tfrac{2}{m(m-1)}\big[ m-1, m-2, \cdots, 1, 0 \big]
\end{align}
 Evaluating $h_{\phi_{m-1}}$ gives us the expression in \eqref{eq:minimum}.
\end{proof}

Now, we continue from \eqref{eq:to_minimize}. We obtain
\begin{align}
  C^{\C}(\mathcal{K}_m,p) &= \max_{\mP_{U}, x} \big( H(Y) - H(Y\mid U) \big)\notag \\
  &\leq \log_2 m - \min_u \min_\phi h_{\phi}\\
  &= \frac{1}{m}\sum_{k=1}^m \big( 1+  pc_{m,k}  \big)\log_2 \big( 1+  pc_{m,k}  \big) \label{eq:apply_minimization}
\end{align}
where in Step \eqref{eq:apply_minimization} we applied Lemma \ref{lem:minimization}.

On the other hand, the upper bound in \eqref{eq:apply_minimization} is attained by choosing $U$ to be uniformly distributed over $\{0,1\}$, and $x(u,s) = u$ for all $u\in \{0,1\}$ and $s \in \mathcal{S}$. \qed

\section{Proof of Corollary \ref{cor:gain_factor}} \label{proof:gain}
In this section, we write $g(\epsilon) = o(f(\epsilon))$ if $\lim_{\epsilon \to 0^+} \frac{g(\epsilon)}{f(\epsilon)} = 0$.

Let $r_m = \tfrac{\cot(\pi/2m)}{2(m-1)}+\tfrac{1}{2}$. Then $R^{\EA}(\mathcal{K}_m,p) = 1-H_b\big( pr_m + \tfrac{1-p}{2} \big) = 1-H_b\big( \tfrac{1}{2} + p(r_m - \tfrac{1}{2}) \big)$.  Since $H_b(\tfrac{1}{2}+\epsilon) = 1-\tfrac{2}{\ln 2} \epsilon^2 + o(\epsilon^2)$, and $p(r_m-\tfrac{1}{2})\to 0$ as $p\to 0$, we obtain that $R^{\EA}(\mathcal{K}_m,p) = \tfrac{2}{\ln 2}\big( \tfrac{\cot(\pi/2m)}{2(m-1)} \big)^2p^2 + o(p^2)$.

On the other hand, $C^{\C}(\mathcal{K}_m,p) = \frac{1}{m}\sum_{k=1}^m \big( 1+  pc_{m,k}  \big)\log_2 \big( 1+  pc_{m,k}  \big)$, where $c_{m,k} = \tfrac{m-2k+1}{m-1}$. Since $\log_2(1+\epsilon) = \tfrac{\epsilon}{\ln 2} - \frac{\epsilon^2}{2\ln 2} + o(\epsilon^2)$, and $pc_{m,k}\to 0$ as $p\to 0$, we obtain that $\log_2 (1+pc_{m,k}) = \tfrac{pc_{m,k}}{\ln 2} - \tfrac{p^2 c_{m,k}^2}{2\ln 2} + o(p^2)$, and that  $(1+pc_{m,k})\log_2(1+pc_{m,k}) =   (1+pc_{m,k}) \big( \tfrac{pc_{m,k}}{\ln 2} - \tfrac{p^2 c_{m,k}^2}{2\ln 2} + o(p^2) \big) = \tfrac{pc_{m,k}}{\ln 2}+\tfrac{p^2 c_{m,k}^2}{2\ln 2} + o(p^2)$. It can be verified that $\sum_{k=1}^m c_{m,k} = 0$ and $\sum_{k=1}^m c_{m,k}^2 = \tfrac{m(m+1)}{3(m-1)}$. It follows that $C^{\C}(\mathcal{K}_m,p) = \tfrac{1}{2\ln 2} \tfrac{m+1}{3(m-1)}p^2 + o(p^2)$.

Therefore, $\lim_{p \to 0^+} \frac{R^{\EA}(\mathcal{K}_m,p)}{C^{\C}(\mathcal{K}_m,p)} = \frac{\tfrac{2}{\ln 2} \big( \tfrac{\cot(\pi/2m)}{2(m-1)} \big)^2}{\tfrac{1}{2\ln 2} \tfrac{m+1}{3(m-1)}} = \frac{3\cot^2 ( \tfrac{\pi}{2m} )}{m^2-1}$.

Finally, $\lim_{m\to \infty}\frac{3\cot^2 ( \tfrac{\pi}{2m} )}{m^2-1} = \lim_{m\to \infty}\frac{3/(m^2-1) }{\tan^2 ( \tfrac{\pi}{2m} )} = \lim_{m\to \infty}\frac{3/(m^2-1) }{ ( \tfrac{\pi}{2m} )^2} = \frac{12}{\pi^2}$. \qed

\section{Proof of Lemma \ref{lem:singleletter_classical}} \label{proof:singleletter_classical}
  The proof is by induction. Note that the base case $k=1$ is assumed to be true. Suppose that $M^{\C}_{k,\opt} = 1$ for some $k\geq 1$. Given any $(M=2,N=k+1)$ (deterministic) classical coding scheme with encoders $\{\phi^{(i)}\}_{i\in[k+1]}$,  there exist $\tilde{s}^k, \tilde{\sigma}^k \in \mathcal{S}^k$ and $\tilde{y}^k \in \mathcal{Y}^k$ such that $\prod_{i=1}^k \mN\big(\tilde{y}_i\mid \phi^{(i)}(w=1, \tilde{s}^i),\tilde{s}_i \big) >0$ and $\prod_{i=1}^k \mN\big(\tilde{y}_i\mid \phi^{(i)}(w=2, \tilde{\sigma}^i), \tilde{\sigma}_i\big) >0$, because otherwise one obtains an $(M=2,n=k)$ zero-error coding scheme by having encoders $\{\phi^{(i)}\}_{i\in[k]}$ and a decoder that depends only on the first $k$ outputs, thus violating the inductive assumption. Similarly, there exist $\tilde{s}_{k+1}, \tilde{\sigma}_{k+1} \in \mathcal{S}$ and $\tilde{y}_{k+1} \in \mathcal{Y}$ for which $\mN\big(\tilde{y}_{k+1} \mid \phi^{(k+1)}(w=1,s^{k}=\tilde{s}^k, s_{k+1}=\tilde{s}_{k+1}),\tilde{s}_{k+1}\big)>0$ and $\mN\big(\tilde{y}_{k+1} \mid \phi^{(k+1)}(w=2,s^{k}=\tilde{\sigma}^k, s_{k+1}=\tilde{\sigma}_{k+1}),\tilde{\sigma}_{k+1}\big)>0$, because otherwise one obtains an $(M=2,n=1)$ zero-error scheme $\tilde{\phi}$ defined by $\tilde{\phi}(w=1,s) \triangleq \phi^{(k+1)}(w=1, s^k=\tilde{s}^k, s_{k+1}=s)$ and $\tilde{\phi}(w=2, s) \triangleq \phi^{(k+1)}(w=2, s^k=\tilde{\sigma}^k, s_{k+1}=s)$ for all $s\in \mathcal{S}$, violating the base case $k=1$. Therefore, there exist $\tilde{s}^{k+1}, \tilde{\sigma}^{k+1} \in \mathcal{S}^{k+1}$ and $\tilde{y}^{k+1}\in \mathcal{Y}^{k+1}$ such that $\prod_{i=1}^{k+1} \mN\big(\tilde{y}_i\mid \phi^{(i)}(w=1, \tilde{s}^i), \tilde{s}_i\big) >0$ and $\prod_{i=1}^{k+1} \mN\big(\tilde{y}_i\mid \phi^{(i)}(w=2, \tilde{\sigma}^i), \tilde{\sigma}_i\big) >0$, implying that the scheme cannot have zero error. Therefore, any $(M=2,n=k+1)$ classical coding scheme cannot have zero error, so $M_{k+1,\opt}^{\C} = 1$.\qed

\section{Proof of Theorem \ref{thm:no_activation_for_graph}} \label{proof:no_activation_graph}
From Lemma \ref{lem:singleletter_classical} it follows that  $c^{\C}_0(\mathcal{G})=0 \Leftrightarrow C^{\C}_0(\mathcal{G})=0$.
In the following, we first prove that $c^{\C}_0(\mathcal{G})=0 \implies$ ``$\mathcal{G}$ is not bipartite", and then prove that ``$\mathcal{G}$ is not bipartite" $\implies c^{\EA}_0(\mathcal{G})=0$. The proof of Theorem \ref{thm:no_activation_for_graph} can then be concluded by noting the obvious direction $c^{\EA}_0(\mathcal{G})=0 \implies c^{\C}_0(\mathcal{G})=0$.
\subsection{$c^{\C}_0(\mathcal{G})=0 \implies$ ``$\mathcal{G}$ is not bipartite"}
It suffices to show that if $\mathcal{G}$ is bipartite, then $c^{\C}_0(\mathcal{G})>0$.
In fact, if $\mathcal{G}$ is bipartite, then there exists an $(M=2,n=1)$ coding scheme which has zero error, so $c^{\C}_0(\mathcal{G}) \geq \log_2(2) = 1$. 
To see this, recall that $\mathcal{G}(\mathcal{Y},\mathcal{S})$ defines the graph channel with state, such that the vertex set $\mathcal{Y}$ corresponds to the channel outputs and the edge set $\mathcal{S}$ corresponds to the channel states. The decoder is defined by a valid bipartition of the graph, i.e., $\mathcal{Y} = \mathcal{Y}_1 \cup \mathcal{Y}_2$ with $\mathcal{Y}_1\cap \mathcal{Y}_2=\emptyset$, such that every edge has one endpoint in $\mathcal{Y}_1$ and the other in $\mathcal{Y}_2$. It then declares $\widehat{W}=1$ when $y\in \mathcal{Y}_1$, and $\widehat{W}=2$ when $y\in \mathcal{Y}_2$. Note that conditioned on each state (edge) $S\in \mathcal{S}$, the encoder controls which endpoint of $S$ is to be seen by the receiver, and since the graph is bipartite, it is free to choose a point in $\mathcal{Y}_1$ for $W=1$ or a point in $\mathcal{Y}_2$ for $W=2$, thus having zero error.  \hfil \qed

\subsection{``$\mathcal{G}$ is not bipartite" $\implies c^{\EA}_0(\mathcal{G}) = 0$}
It remains to prove that if $\mathcal{G}$ is not bipartite, then $c^{\EA}_0(\mathcal{G}) = 0$.  
It is well known that a if $\mathcal{G}$ is not bipartite, then $\mathcal{G}$ contains a cycle of odd length $m$. 
Without loss of generality, the vertices in this cycle are labeled as $Y\in \{0,1,\cdots, m-1\}$, and the edges in the cycle are labeled as $S\in \{0,1,\cdots, m-1\}$, so that $Y= S-X \pmod m$ for $X\in \{0,1\}$.

In the following, for a positive semi-definite operator $\rho\in \mathbb{C}^d$, we denote $\supp(\rho)$ as the support of $\rho$, i.e., the span of the eigenvectors corresponding to its nonzero eigenvalues. We use $\Pi_{\rho}$ to denote the orthogonal projector onto $\supp(\rho)$. For two positive semi-definite operators $\rho, \sigma \in \mathbb{C}^d$, we write $\sigma \succeq \rho$ if $\sigma-\rho$ is positive semi-definite, and we write $\rho \to \sigma$ if $\Pi_{\rho} \sigma \Pi_{\rho} \succeq \rho$.
The following lemma will be useful.
\begin{lemma}\label{lem:inclusion}
  Let $\rho_0,\rho_1, \sigma_0$, and $\sigma_1$ be positive semi-definite operators in $\mathbb{C}^d$ such that $\rho_0+\rho_1 = \sigma_0 + \sigma_1$ and $\rho_0 \perp \sigma_1$. Then, $\rho_0 \to \sigma_0$ and $\sigma_1 \to \rho_1$.
\end{lemma}
\begin{proof}
$\Pi_{\rho_0} \sigma_0 \Pi_{\rho_0} \stackrel{(a)}{=} \Pi_{\rho_0} (\sigma_0 + \sigma_1) \Pi_{\rho_0} \stackrel{(b)}{=}  \Pi_{\rho_0} (\rho_0 + \rho_1) \Pi_{\rho_0} =  \rho_0 + \Pi_{\rho_0}\rho_1 \Pi_{\rho_0} \stackrel{(c)}{\succeq} \rho_0$.  Step (a) is because $\rho_0 \perp \sigma_1$, so $\Pi_{\rho_0} \sigma_1 \Pi_{\rho_0} = {\bf 0}$. Step (b) uses the condition $\rho_0+\rho_1 = \sigma_0 + \sigma_1$. Step (c) is because $\Pi_{\rho_0} \rho_1 \Pi_{\rho_0}$ is positive semi-definite.
\end{proof}

To set up a proof by contradiction, suppose $c_0^{\EA}(\mathcal{G}) > 0$, then there exists an entanglement-assisted $(M=2,n=1)$ coding scheme with zero error. Let $\alpha_{x\mid s} \triangleq \Tr_A\big[\rho_{x\mid w=1,s}\big]$ and let $\beta_{x\mid s} \triangleq \Tr_A\big[\rho_{x\mid w=2,s}\big]$. 
Then,
\begin{align} 
  \alpha_{0\mid s}+\alpha_{1\mid s}&=  \Tr_A\big[ \sum_{x\in\{0,1\}} \rho_{x\mid w=1,s} \big]\\
  &=\Tr_A\big[(\sum_{x\in \{0,1\}}\mathcal{E}^{(1)}_{x\mid w=1,s} \otimes {\rm id}_B)(\rho_{AB})\big]  \\
  &= \rho_B \label{eq:identical_sum}
\end{align}
regardless of $s$ for $s\in \mathcal{S}$, where $\rho_B \triangleq \Tr_A\big[\rho_{AB}\big]$ is the initial reduced state at the receiver. 
Similarly 
\begin{align}
	\beta_{0\mid s}+\beta_{1\mid s} = \rho_B, \label{eq:identical_sum_2}
\end{align}
regardless of $s$. This reflects the fact that the transmitter's local measurement will not change the receiver's (average) quantum state, also known as the no-signaling (no-communication) theorem (e.g., \cite{Peres_QIT}). Meanwhile, for each $x,x'\in \{0,1\}, s,s'\in \{0,1,\cdots, m-1\}$ such that $s-x \equiv s'-x' \mod m$, we must have
\begin{align} \label{eq:orthogonal}
  \alpha_{x\mid s} \perp \beta_{x'\mid s'}.
\end{align}
To see this, note that the two cases, namely, $(W=1, S=s, X=x)$ and $(W=2, S=s', X=x')$ lead to the same channel output $Y$, so the receiver has to distinguish between $W=1$ and $W=2$ by measuring the quantum state, which is either in $\alpha_{x\mid s}$ (if $W=1$), or $\beta_{x'\mid s'}$ (if $W=2$). In order for the two states to be perfectly distinguishable, \eqref{eq:orthogonal} must hold.

Fig. \ref{fig:chain} shows an implied chain of $\to$ relations.
\begin{figure}[htbp]
\center
\begin{tikzpicture}
  \node (E1) {\footnotesize $\alpha_{0|0} \to \beta_{0|1} \to \alpha_{0|2} ~\cdots~  \alpha_{0|m\!-\!1} \to \beta_{0|0} \to \alpha_{0|1}  ~\cdots~  \beta_{0|m\!-\!1} \to \alpha_{0|0}$};
  \node[below=0.5cm of E1] (E2) {\footnotesize $\alpha_{1|0} \gets \beta_{1|1} \gets \alpha_{1|2} ~\cdots~  \alpha_{1|m\!-\!1} \gets \beta_{1|0} \gets \alpha_{1|1}  ~\cdots~  \beta_{1|m\!-\!1} \gets \alpha_{1|0}$};
  \draw[thick, black] (-4,-0.25) -- node[midway, fill=white, inner sep=1pt] {\footnotesize $\perp$} (-3.2,-0.75);
  \draw[thick, black] (-3,-0.25) -- node[midway, fill=white, inner sep=1pt] {\footnotesize $\perp$} (-2.2,-0.75);
  \draw[thick, black] (-0.5,-0.25) -- node[midway, fill=white, inner sep=1pt] {\footnotesize $\perp$} (0.3,-0.75);
  \draw[thick, black] (0.5,-0.25) -- node[midway, fill=white, inner sep=1pt] {\footnotesize $\perp$} (1.3,-0.75);
  \draw[thick, black] (3,-0.25) -- node[midway, fill=white, inner sep=1pt] {\footnotesize $\perp$} (3.8,-0.75);
  \node at (-1.5,-0.5) {\small $\cdots$};
  \node at (2,-0.5) {\small $\cdots$};
\end{tikzpicture}
\caption{Chain of relations implied by \eqref{eq:identical_sum}, \eqref{eq:identical_sum_2}, \eqref{eq:orthogonal} and Lemma \ref{lem:inclusion}.}\label{fig:chain}
\end{figure}
To see this, note that for each column, there are two operators, and their sum is equal to $\rho_B$ according to \eqref{eq:identical_sum} and \eqref{eq:identical_sum_2}. Also, \eqref{eq:orthogonal} implies the $\perp$ relation between some operators shown across the two rows. The two chains of $\to$ relation then follow from Lemma \ref{lem:inclusion}. 

Note that the relation $\rho \to  \sigma$ means $\Pi_{\rho}\sigma \Pi_{\rho} \succeq \rho$, so $\Tr\big[ \rho \big] \leq \Tr\big[ \Pi_{\rho}\sigma \Pi_{\rho} \big] = \Tr\big[ \Pi_{\rho}^2 \sigma \big] = \Tr\big[ \Pi_{\rho} \sigma  \big] \leq \Tr\big[ \Pi_{\rho} \sigma  \big] +  \Tr\big[ (I-\Pi_{\rho}) \sigma  \big] = \Tr\big[ \sigma \big]$.
Applying the argument recursively based on the top chain of $\to$ relations, we obtain that
\begin{equation} \label{eq:chain}
\begin{aligned}
	 \Tr\big[\alpha_{0|0}\big] &\leq \Tr\big[ \Pi_{\alpha_{0|0}} \beta_{0|1} \big] \leq  \Tr\big[ \beta_{0|1} \big] \\
	 &\leq \Tr\big[ \Pi_{\beta_{0|1}} \alpha_{0|2}  \big] 	\leq  \Tr\big[ \alpha_{0|2} \big] \\
	 &~~\vdots \\
     &\leq   \Tr\big[ \Pi_{\beta_{0|m-1}} \alpha_{0|0}  \big] \leq \Tr\big[\alpha_{0|0}\big]
\end{aligned}
\end{equation}
Therefore, all traces in \eqref{eq:chain} must be equal. In particular, $\Tr\big[\alpha_{0|0}\big] = \Tr\big[ \Pi_{\beta_{0|m-1}}  \alpha_{0|0} \big]$, so $\Tr\big[(I-\Pi_{\beta_{0|m-1}}) \alpha_{0|0}\big] = 0$, and thus $\alpha_{0|0}=\Pi_{\beta_{0|m-1}} \alpha_{0|0}$. It follows that $\supp(\alpha_{0|0}) \subseteq \supp(\beta_{0|m-1})$, because any element in $\supp(\alpha_{0|0})$ can be written as (for some $v$) $\alpha_{0|0} v = \Pi_{\beta_{0|m-1}} \alpha_{0|0}v \in \supp(\beta_{0|m-1})$. Similarly,  $\supp(\beta_{0|m-1}) \subseteq \supp(\alpha_{0|m-2}) \subseteq \cdots \subseteq \supp(\alpha_{0|0})$, so all supports $\supp(\alpha_{0|0}), \supp(\beta_{0|1}), \cdots,$ $ \supp(\beta_{0|m-1})$ are equal.  In particular, $\supp(\alpha_{0|0}) = \supp(\beta_{0|0})$. However, \eqref{eq:orthogonal} further implies that $\alpha_{0|0} \perp \beta_{0|0}$, so $\supp(\alpha_{0|0})=\supp(\beta_{0|0}) = \{0\}$, i.e., $\alpha_{0|0}=\beta_{0|0}={\bf 0}$. Due to similar reasons, the second row implies $\alpha_{1|0} = \beta_{1|0} = {\bf 0}$, so $\alpha_{0|0} + \alpha_{1|0} = {\bf 0}$. Note that \eqref{eq:identical_sum} implies that $\alpha_{0|0} + \alpha_{1|0} = \rho_B$. Since $\rho_B$ is a (normalized) density operator, $\rho_B \not={\bf 0}$. We thus arrive at a contradiction.  \hfil \qed

\section*{Acknowledgement}
The authors acknowledge the use of ChatGPT to assist in the derivations in \eqref{eq:use_prob1}--\eqref{eq:interm2} and in obtaining a shorter proof of Corollary \ref{cor:gain_factor} in Appendix B. All arguments and content were carefully reviewed and edited by the authors.

\bibliographystyle{IEEEtran}
\bibliography{../../bib_file/yy.bib}
\end{document}